\documentclass[format=acmsmall]{acmart}

\pdfoutput=1
\acmJournal{TECS}
\acmYear{2017}
\acmMonth{10}
\copyrightyear{2017}

\setcopyright{acmcopyright}

\received{April 2017}
\received[revised]{June 2017}
\received[accepted]{June 2017}

\acmPrice{15.00}

\usepackage{epstopdf}
\usepackage{dblfloatfix}
\usepackage{booktabs}
\usepackage{enumitem}
\usepackage{url}
\usepackage[normalem]{ulem}
\usepackage{subcaption}
\usepackage{color}
\usepackage{adjustbox}
\usepackage{caption} 
\usepackage{wrapfig}
\usepackage{framed}
\newcommand{\T}{\tilde{T}}

\begin{document}




\title{Power-Temperature Stability and Safety Analysis for Multiprocessor
Systems}
\author{Ganapati Bhat}
\orcid{0000-0003-1085-2189}
\affiliation{%
	\institution{Arizona State University}
	\department{School of Electrical, Computer and Energy Engineering}
	\city{Tempe}
	\state{AZ}
	\postcode{85287}
	\country{USA}
}
\author{Suat Gumussoy}
\affiliation{%
	\institution{IEEE Member}
}
\author{Umit Y. Ogras}
\affiliation{%
	\institution{Arizona State University}
	\department{School of Electrical, Computer and Energy Engineering}
	\city{Tempe}
	\state{AZ}
	\postcode{85287}
	\country{USA}
}

\begin{abstract}
Modern multiprocessor system-on-chips (SoCs) integrate multiple heterogeneous 
cores to achieve high energy efficiency.
The power consumption of each core contributes to an increase in the temperature across the chip floorplan.
In turn, higher temperature increases the leakage power exponentially, 
and leads to a positive feedback with nonlinear dynamics. 
This paper presents a power-temperature stability and safety analysis technique for multiprocessor systems.
This analysis reveals the conditions under which the power-temperature trajectory converges to a stable fixed point. 
We also present a simple formula to compute the stable fixed point and maximum 
thermally-safe 
power consumption at \textit{runtime}. 
Hardware measurements on a state-of-the-art mobile processor show that our 
analytical formulation can predict the stable fixed point with an average error 
of 2.6\%.
Hence, our approach can be used at runtime to ensure thermally safe operation and guard against thermal threats.
\end{abstract}

%
%

%
%



\thanks{This article was presented in the International Conference on 
Hardware/Software Codesign
	and System Synthesis (CODES+ISSS) 2017 and appears as part of the 
	ESWEEK-TECS special issue.
	
This work was supported partially by Semiconductor Research Corporation (SRC) 
task 2721.001 and National Science Foundation grant CNS-1526562.
	
	Author's addresses: G. Bhat {and} U. Y. Ogras, School of Electrical, 
	Computer and Energy Engineering, Arizona State University, Tempe, AZ, 
	85287; emails: \{gmbhat, umit\}@asu.edu; S. Gumussoy, Boston, MA, USA; 
	email: suat@gumussoy.net.}

\maketitle

\section{Introduction}

Power consumption and the resulting heat dissipation are among the major 
problems faced by mobile platforms.
Besides draining the battery,
high temperature deteriorates the reliability and user experience~\cite{huang2006hotspot,egilmez2015user}.
Recent evidence also shows that uncontrollable temperature increase
in one part of the chip poses serious safety risks~\cite{note7recall}.
To mitigate these risks, commercial chips typically have a hard-coded maximum safe temperature.
Thermal sensors 
monitor the temperature at critical hotspots.
If the observed temperature exceeds the maximum limit,
the system throttles the computational resources or shuts down the platform
depending on the severity of the violation.
However, these techniques are triggered only after the damage becomes 
observable.

There is a well-known positive feedback between the power consumption and the
temperature~\cite{vassighi2006thermal,liao2003coupled}.
Power consumption drives the chip temperature up through thermal resistance and capacitance networks.
Higher temperature, in turn, leads to an exponential increase in leakage power.
This nonlinear dynamics leads to a positive feedback which increases both the 
temperature and power consumption. 
\textit{When a stable fixed point exists}, it attracts all the temperature 
trajectories 
within the region of convergence to itself. 
Therefore, the steady increase in power consumption and temperature continues 
until the stable fixed point is reached. 
Otherwise, a thermal runaway occurs, as we prove in this paper.
%

Although the consequences of thermal runaway are detrimental, 
to the best of our knowledge, 
there are no models that can analyze the existence and stability of fixed points at runtime. 
More importantly, even if a stable fixed point exists,
it may be well above the maximum safe temperature limit. 
A static bound on the maximum power consumption can neither avoid thermal 
violations, 
nor provide insight on the expected time before a thermal violation occurs.
Therefore, it is critical to monitor the stability and safety of power-temperature dynamics at \textit{runtime},
and detect potential violations \textit{before any damage is done}.


\begin{figure}[t]
    \vspace{-1mm}	
    \centering
	\includegraphics[width=0.6\linewidth]{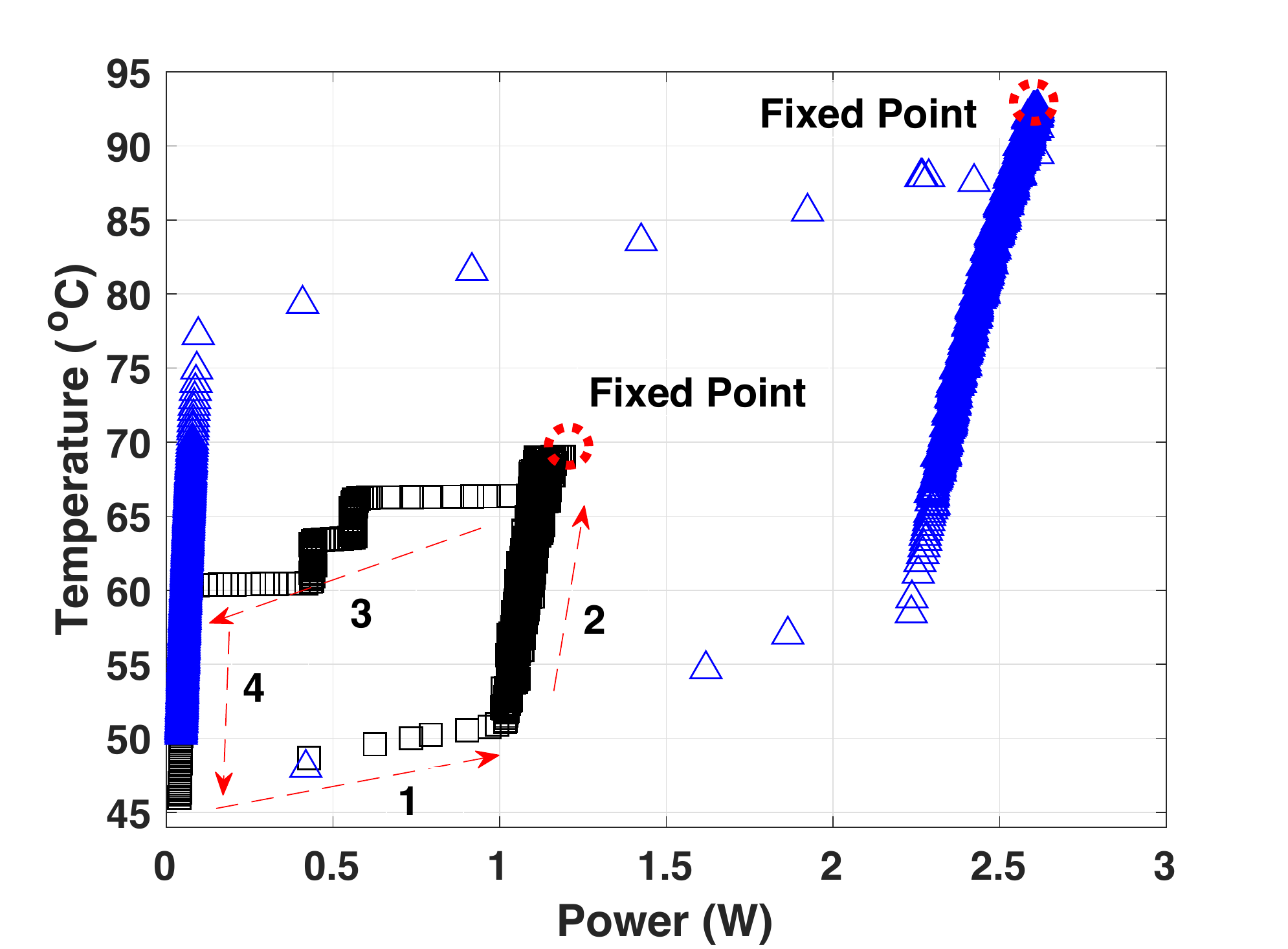}
	\vspace{-1mm}
	\caption{
			Illustration of the power consumption temperature trajectory for two power consumption levels using
	experiments performed on Odroid XU3 board~\cite{ODROID_Platforms}.} 
    \vspace{-2mm}
	\label{fig:motivation}
\end{figure}

To illustrate the problem addressed by this paper, 
we measured the power-temperature profile of a commercial SoC 
for a complete heat-up/cool-down cycle, as shown in Figure~\ref{fig:motivation}.
In particular, the inner trajectory, denoted by black $\square$ markers, starts 
with a power consumption slightly larger than 0.2~W.
When we increase the dynamic activity, the power consumption quickly rises to 0.9~W. 
As a result, the temperature starts ramping up during this period marked as ``1''.  
Then, we keep the dynamic activity constant, but  
the temperature continues to increase towards a fixed point (segment 2).
The corresponding rise in the power consumption reveals the impact of the leakage power,
since the dynamic activity is kept constant.
Eventually, the power consumption and temperature converge to (1.2~W, 
69$^\circ$C)\footnote{
After the dynamic activity is reduced, first, the power consumption drops. 
Then, the temperature starts decreasing (segments 3 and 4).}.
The existence of this fixed point (i.e., the upper right corner) is necessary to avoid a thermal runaway, 
but it is \textit{not sufficient} to ensure a thermally safe operation.
For example, when we repeat the same experiment with a higher dynamic activity,
we observe the other trajectory denoted by $\triangle$ markers.
The second trajectory also converges to a stable fixed point given by 
(2.6~W, 
93$^\circ$C).
However, this point is larger than the temperature which triggers throttling.
If a demanding application or a power virus drive the system to this fixed 
point, throttling can deteriorate the performance.
In contrast, thermally safe operation without performance penalties can be achieved, 
if we can compute the fixed point and the expected time to reach it as a 
function of the dynamic activity.

\vspace{1mm}
%
%

\noindent\textbf{The major contributions of this paper are as follows:}

\begin{enumerate} [leftmargin=*]

\item We first show that the power-temperature dynamics have 
either no fixed point or two fixed points, 
as a function of the system parameters and the dynamic power 
consumption~(Section~\ref{sec:fixed_point_condition}). 

\item We prove that the no-fixed-point case is unstable and 
causes thermal 
runaway~(Section~\ref{sec:stability}).

\item When there are two fixed points, we prove that one of these fixed points
is stable and we give the region of convergence, i.e., the temperature 
interval where any temperature inside it converges to the stable fixed point. 
We also prove that the second fixed point is unstable. We derive the region of 
convergence and the intervals for which the temperature 
diverges~(Section~\ref{sec:stability}).

\item We derive an analytical formula to compute the maximum dynamic power consumption allowed
to guarantee that the temperature does not exceed a thermally safe value (Section~\ref{sec:computation}). 

\item To validate the proposed approach, 
we present thorough experimental evaluations on an 8-core big.LITTLE platform ~\cite{ODROID_Platforms} 
using single-threaded, multi-threaded and concurrent applications. 
We demonstrate that the average and maximum prediction errors are 
2.6\% and 6.2\%, respectively. 
We also show that the total computational overhead of all
proposed computations is 75.2~$\mu$s of 100 ms control interval, 
i.e., $\approx$ 0.075\% (Section~\ref{sec:experiments}).

\end{enumerate}

\noindent\textbf{Potential Impact:}
This paper lays the theoretical foundation for power-temperature stability 
analysis and presents experimental validation of the contributions
summarized in the enumerated list above. 
As we demonstrate in Section~\ref{sec:experiments}, 
our power-temperature stability analysis has a very efficient and practical implementation 
despite the complexity of the derivations. 
Therefore, it can be employed by other researchers in dynamic thermal and 
power management~(DTPM) algorithms to determine if the power-temperature 
dynamics is stable or not. If any instability is detected, immediate corrective 
actions can be taken by the system.
Otherwise, the proposed approach can be used to predict the stable fixed point 
and the expected time to reach it. 
DTPM algorithms can use this prediction to determine the urgency and degree of the response. 
Finally, the proposed approach can accurately compute the maximum power consumption that can be
tolerated before violating the thermal constraints.
This insight can be used by DTPM algorithms to make informed 
decisions. For example, if a power-hungry application is 
driving the system 
beyond a safe temperature, the DTPM algorithm can selectively isolate the 
application or terminate it.
The proposed approach can also be applied as an effective built-in test 
to determine reliability and thermal violation risks.

The rest of this paper is organized as follows.
We present the related work in Section~\ref{sec:related_work}.
We give an overview of the proposed methodology and detail the theoretical 
derivations
in Section~\ref{sec:overview} and Section~\ref{sec:fixed_point_prediction}, respectively.
Finally, we present the experimental results in Section~\ref{sec:experiments},
and summarize our conclusions in Section~\ref{sec:conclusion}. 
\section{Related Work and Novelty} \label{sec:related_work}
Thermal modeling and analysis have recently attracted significant attention 
due to large power densities and the impact of temperature on 
reliability~\cite{vassighi2006thermal,li2004efficient}. 
These studies can be broadly classified as design time and runtime 
approaches. 
Design time approaches primarily focus on a full-chip thermal analysis such 
that parameters like thermal design power can be 
determined~\cite{huang2006hotspot,huang2009full,yang2007isac,zhan2005high}. For 
instance, 
Hotspot~\cite{huang2006hotspot} models the thermal 
behavior of the entire chip as a function of the floorplan, technological 
parameters and packaging.
Then, power consumption traces obtained using common benchmarks are used to 
simulate the thermal behavior. 
Similarly, authors in~\cite{yang2007isac} propose a tool which does the 
full-chip thermal analysis during the synthesis of a chip.
These models are very useful for early design stages,
however, the high fidelity of these approaches comes at the expense of 
computational complexity. 
It is possible to extract high-level models from these tools, 
and evaluate them iteratively to analyze the thermal behavior. 
However, iterative approaches are time-consuming and not accurate as we 
demonstrate in Section~\ref{sec:implementation_overhead}.

Virtually all commercial products have a mechanism to throttle performance, 
or shut down the whole system in case of thermal violations.
However, reactive approaches penalize performance, 
and respond only after the fact~\cite{sahin2016qscale, 
	isci2006analysis}.
This led to predictive approaches for dynamic thermal and power management.
Predictive approaches first develop computationally efficient thermal 
models which can be used at runtime~\cite{beneventi2014effective, 
hanumaiah2011performance}. 
These models are  used to predict the temperature as a function 
of the power consumption 
to guide the DTPM algorithms~\cite{prakash2016improving,cochran2013thermal,singla2015predictive,kumar2008system}.
For instance, the authors in~\cite{kumar2008system} propose a hybrid thermal 
management algorithm which uses hardware and software techniques for temperature control. 
In particular, they employ clock gating and thermal-aware scheduling 
to improve the performance of the algorithm. 
Similarly, the work presented in~\cite{xie2013dynamic} characterizes 
the thermal system parameters offline, considering the coupling between various 
components. 
Then, this model is used at runtime to predict the temperature, 
and control the frequency to minimize thermal violations.
The authors in~\cite{liu2013dynamic,brooks2007power} propose methods 
that consider the transient thermal effects and use them for thermal management~\cite{liu2013dynamic}.
While these models work for short prediction intervals, the error increases 
considerably when larger prediction windows are 
used~\cite{singla2015predictive}. 
Furthermore, they do not analyze the existence and stability of thermal fixed points. 
In contrast, our approach can accurately estimate the fixed point and maximum 
allowed power consumption 
at runtime with a low computational overhead. 
Therefore, it can be utilized by DTPM algorithms.

\vspace{1mm}
Recent studies have also proposed techniques to calculate a thermally sustainable power 
budget~\cite{pagani2017thermal,chen2016tsocket}, 
and maintain it at runtime~\cite{gupta2017dynamic}. 
In particular, the authors in~\cite{pagani2017thermal} propose a method to calculate a 
thermally safe power such that thermal constraints of the system are not violated. 
However, it does not consider the positive feedback between 
leakage power and temperature.
The work in~\cite{chen2016tsocket} proposes a framework called TSocket which 
evaluates the sustainable power budget for different threading strategies in a 
multiprocessor system. 
These studies do not address the problem of calculating the thermal fixed 
point and the conditions on existence of a fixed point.
They also employ mainly simulation tools such as HotSpot. 
In contrast, our approach of finding the maximum safe power is implemented and 
validated on a real hardware platform.

\vspace{1mm}
A number of studies analyze the positive feedback effect between power 
consumption and 
temperature~\cite{liao2003coupled,heo2003reducing,vassighi2006thermal}. 
In particular, the authors of~\cite{liao2003coupled} show that a thermal 
runaway is implied 
when the second order derivative of temperature with respect to time is 
positive. 
As pointed out by the authors, this criterion can be successfully applied 
during 
design time analysis and simulation. 
However, it cannot be used as a preventive measure at runtime, 
since it is satisfied only after the thermal runaway kicks off. 
Similarly, the technique presented in~\cite{vassighi2006thermal} uses a simple 
junction-to-ambient heat removal model 
to predict a thermal runaway during burn-in reliability screening before 
shipping the chip.
The authors in~\cite{heo2003reducing} use the temperature dependence of leakage 
to increase thermally sustainable power dissipation through activity migration.
However, neither one of these techniques addresses the problem of 
calculating the fixed point when there is no thermal runaway.
Our work addresses this problem by first deriving the conditions for the 
existence of a fixed point. When the fixed point exists, we provide the region 
of convergence for the power-temperature dynamics. Then, we predict the stable 
fixed point of the system.
Hence, it can be used to guard against power attacks that aim to induce damage 
by elevating the 
temperature~\cite{dadvar2005potential,hasan2005heat,kong2010thermal}.
\section{Preliminaries and Overview} \label{sec:overview}
This section first presents the power consumption and temperature models
required for the proposed analysis. 
Readers familiar with these models can jump to Section~\ref{sec:challenges},  
where we  summarize the challenges and give an overview of the proposed 
approach.

\subsection{Power and Temperature Models} \label{sec:models}
Suppose that there are $M$ processors in the target system, as summarized in 
Table~\ref{tab:summary_symbols}.
We can express the power consumption of processor $i$
as the sum of dynamic and leakage power consumption:
\begin{equation}\label{eqn:dynamic_power}
P_{i} = C_{sw,i} V_{i}^2 f_i + V_{i}I_{leak,i}
\end{equation}
where $C_{sw,i}$ is the switching capacitance, $V_i$ is the supply voltage and $f_i$ is the operating frequency.
The leakage current $I_{leak,i}$, which depends on the 
temperature $T$,
can be approximated as the sum of the gate leakage and subthreshold current as:
%
\begin{equation} \label{eqn:leakge_current}
I_{leak,i} = I_{g,i} + A_s\frac{W_i}{L} \frac{kT^2}{q} e^\frac{q(V_{GS,i} - 
V_{th,i})}{nkT}
\end{equation}
where $I_{g,i}$ is the gate leakage, $A_s$ is a technology constant,
$W_i/L$ is the ratio of the effective channel width to channel length,
$k$ is the Boltzmann constant, $q$ is the electron charge, $V_{GS,i}$ is the gate to source voltage, $V_{th,i}$ is the threshold voltage,
and $n$ is the sub-threshold swing coefficient~\cite{kim2003leakage}.
For notational convenience, we consolidate the technology and design parameters
by introducing the following constants:
\begin{equation} \label{eqn:constants}
\kappa_{1,i} = A_s\frac{W_i}{L}\frac{k}{q},~~~~
\kappa_{2,i} = \frac{q(V_{GS,i} - V_{th,i})}{nk}
\end{equation}
Note that $\kappa_{1,i} >0$, while $\kappa_{2,i}< 0$
since $V_{GS,i} < V_{th,i}$ for sub-threshold voltages.
In summary, the power consumptions of the processors in the target system
can be denoted by the $M \times 1$ vector $\mathbf{P} = [P_1, P_2, \ldots, P_M]^\intercal$,
where $P_i$ is obtained
using Equations~\ref{eqn:dynamic_power}-\ref{eqn:constants}
as:
%
\begin{equation} \label{eqn:total_power}
P_i = C_{sw,i} V_{i}^2 f_i + V_{i}I_{g,i} +  V_{i}\kappa_{1,i} T_i^2 e^\frac{\kappa_{2,i}}{T_i},~~1 \leq i \leq M
\end{equation}

Suppose that there are $N$ thermal hotspots of interest.
The dynamics of the temperature can be expressed using
the power consumption vector $\mathbf{P}$ and
thermal capacitance and conductance matrices~\cite{huang2006hotspot,sharifi2013prometheus}.
We employ the following discrete-time state-space system to 
model 
the temperature,
since the measurements and control decisions are made in periodic intervals:
\begin{equation}\label{eqn:temp_model}
\mathbf{T[k+1]} = A \mathbf{T[k]}+ B \mathbf{P[k]}
\end{equation}
This equation expresses the temperature of $N$ cores 
as a function of the power consumption of $M$ sources,  
e.g., the big/little CPU clusters, GPU and memory. 
The ${N \times N}$ matrix $A$ describes the effect of temperature in time step $k$ 
on the temperature in the next time step.
The ${N \times M}$ matrix $B$, on the other hand, describes the effect of each 
power source on the temperature in the next time step.
When we plug the $M \times 1$ power consumption vector $\mathbf{P}$
from Equation~\ref{eqn:total_power} into Equation~\ref{eqn:temp_model},
we obtain the following system of nonlinear equations:

\noindent
\begin{eqnarray} \label{eqn:nonlinear_model}
\mathbf{T[k+1]} = A\mathbf{T[k]} +
B \big[P_1[k], P_2[k], \ldots, P_M[k] \hspace{0mm}
	\big]^\intercal,~\mathrm{where} \\ \vspace{-2mm}
P_i[k] = C_{sw,i} V_{i}^2 f_i + I_{g,i} V_{i} +  V_i\kappa_{1,i} T_i[k]^2
	e^\frac{\kappa_{2,i}}{T_i[k]},~~~~~~
	\hspace{-0.25mm}~1 \hspace{-0.25mm} \leq \hspace{-0.25mm}i
	\hspace{-0.25mm}\leq \hspace{-0.25mm} M \nonumber
\end{eqnarray}
\begin{table}[t]
	\centering
	\small
	\caption{Summary of major parameters}
	\vspace{-4mm}
	\label{tab:summary_symbols}
	\begin{tabular}{@{}ll@{}}
		\toprule
		Symbol                      & 
		Description                                                             
		
		\\
		\midrule
		$N,M$                         & \begin{tabular}[c]{@{}l@{}} The number 
			of 
			thermal hotspots and processing \\
			elements (resources) in the SoC, respectively. 
		\end{tabular}                                                       
		\\
		\midrule
		$\mathbf{T[k]}$                      & \begin{tabular}[c]{@{}l@{}}$N 
			\times 1$ 
			array where $T_i[k],~1 \leq j \leq N$\\ denotes the temperature of 
			the 
			$i^{th}$ hotspot.\end{tabular}        \\ \midrule
		$T$                         & \begin{tabular}[c]{@{}l@{}} 
			The maximum (scalar) steady state temperature \\
			over all thermal hotspots.\end{tabular}                       \\ 
		\midrule
		$T^*, P_C^*$                         & \begin{tabular}[c]{@{}l@{}} The 
			maximum 
			thermally safe temperature and \\ power, 
			respectively.\end{tabular}                       \\ \midrule
		$\mathbf{P[k]}$                      & \begin{tabular}[c]{@{}l@{}}$M 
		\times 1$ 
			array where $P_i[k],~1 \leq i \leq M$ denotes\\ the power 
			consumption 
			of the $i^{th}$ resource.\end{tabular} \\ \midrule
		$\kappa_{1,i}, \kappa_{2,i}$ & \begin{tabular}[c]{@{}l@{}}Technology 
			dependent parameters of the\\ leakage power for the $i^{th}$ 
			resource.\end{tabular}                            \\ \midrule
		$a, b$                        & \begin{tabular}[c]{@{}l@{}}Parameters 
		of 
			the single input single output\\ model which describes the thermal 
			dynamics.\end{tabular}  \\ \midrule
		$\T, \alpha, \beta$   & Auxiliary parameters introduced in 
		Equation~\ref{change_of_parameters}. 
		\\ 	\bottomrule
	\end{tabular}
\end{table}

\subsection{Challenges and Needs Assessment}  \label{sec:challenges}

The nonlinear system given in Equation~\ref{eqn:nonlinear_model}
shows the positive feedback between 
the power consumption of $M$ processors and $N$ temperature hotspots.
Solving this problem at runtime presents \emph{three major challenges}.
First, a stable solution may not even exist due to the nonlinear positive feedback.
Second, the power consumption and temperature at time $k$
depend on their values at time $k-1$
due to the temperature dependence of the leakage power.
This dependency requires an iterative solution,
which is intractable for runtime analysis, as shown in 
Section~\ref{sec:implementation_overhead}.
Finally, we need to find the maximum power consumption $P_i^*$
that \emph{guarantees} a thermally safe operation.
A simple iteration can find temperature using the power inputs, 
but the opposite direction requires a rigorous approach.

The proposed approach addresses each of these challenges one 
by one as outlined in Figure~\ref{fig:overview_figure}.
We first determine whether a stable fixed point exists 
using the current power and temperature measurements. 
Then, we describe how the stable fixed points can be computed efficiently.
Finally, we derive compact analytical formulae to compute 
the time to reach the fixed point and the constraints on the maximum power consumption $P_i^*$
to avoid thermal violations.
The proposed approach is called with the default frequency governors, 
which typically have a period of 100 ms.

\begin{figure}[h!]
	\centering
	\includegraphics[width=0.7\linewidth]{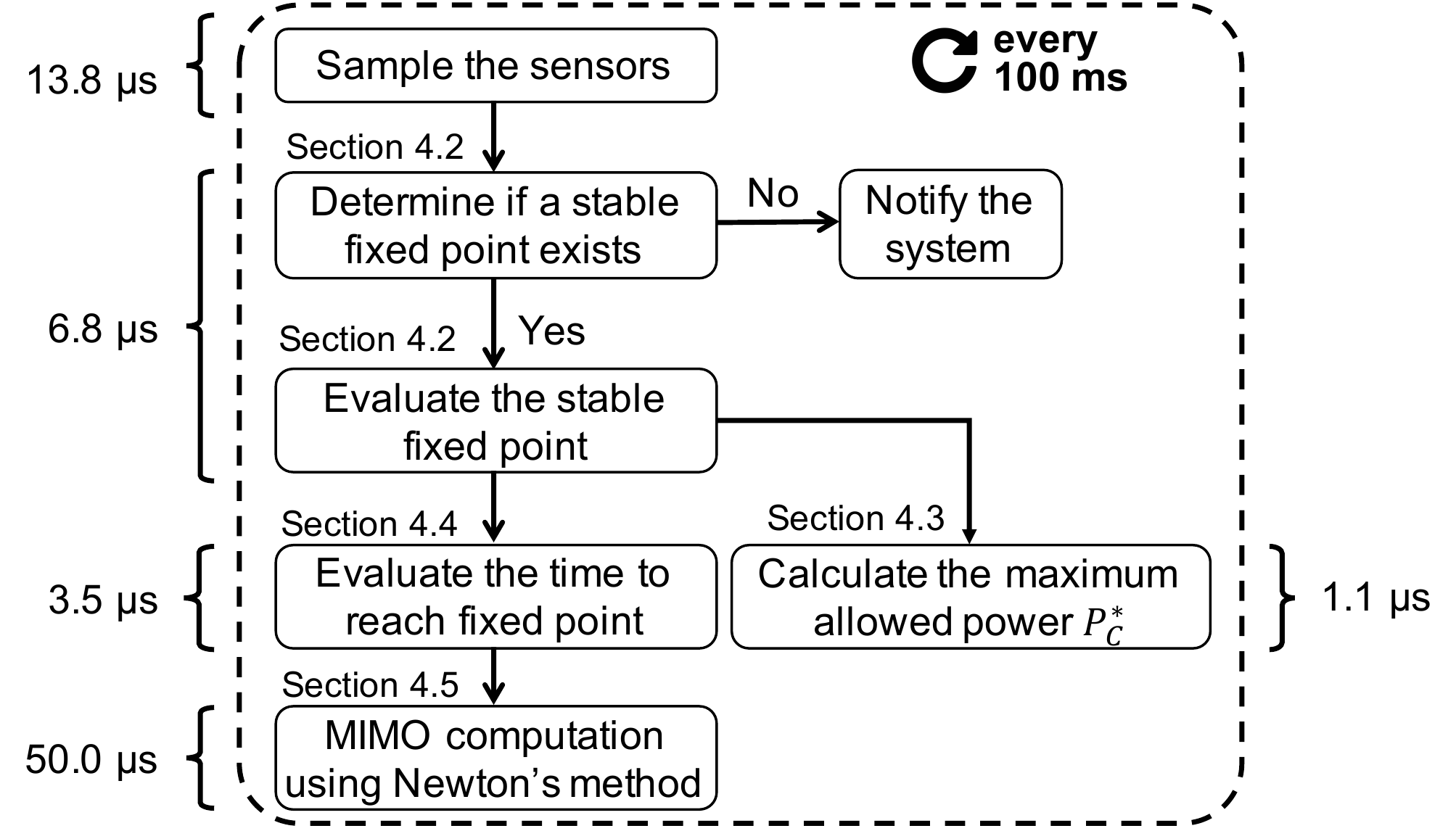}
	\vspace{-3mm}
	\caption{Overview of the proposed approach.}
	\label{fig:overview_figure}
	\vspace{-3mm}
\end{figure}

\section{Thermal Fixed Point Analysis} \label{sec:fixed_point_prediction}
Let the scalar $T$ denote the maximum steady state temperature 
over all thermal hotspots, i.e.:
\begin{equation*}
T = \max_{1 \leq i \leq N} \lim_{k \rightarrow \infty} T_i[k] 
\end{equation*}
Since the thermal safety is determined by the maximum 
temperature, 
we focus on the hotspot with the highest temperature.
At steady state ($k\rightarrow \infty$), we can model the temperature of \emph{each hotspot}
using the following single input single output (SISO) system: 
\begin{equation}\label{eqn:temp_ss}
T = aT + bP
\end{equation}

\noindent where $0<a<1$ and $b>0$ are the parameters of the reduced order system. 
We emphasize that the reduced order model is obtained through system identification,
as described in Section~\ref{sec:sys_id}.
Using a SISO model does not mean that we consider only one core. 
Unlike a crude approximation that directly uses the corresponding entries 
in $A$ and $B$ matrices, 
the coefficient $a$ in our model reflects the thermal coupling between 
different hotspots, 
and coefficient $b$ reflects the impact of multiple power sources. 
We employ a SISO model, since it enables an in-depth theoretical analysis 
with powerful insights for practical scenarios.
We discuss the solution to the multi-input multi-output (MIMO) 
case at the end 
of this section.

To isolate the impact of the temperature,
we rewrite the total power consumption given in Equation~\ref{eqn:total_power} 
as:
\begin{equation}\label{eqn:power_simplifed}
P = P_C + V \kappa_{1}T^2e^{\frac{\kappa_{2}}{T}}
\end{equation}
where $P_C = C_{sw} V^2 f + V I_{g}$ represents the 
\emph{temperature-independent} component,
and subscript $i$ is dropped to simplify the notation. 
Substituting Equation~\ref{eqn:power_simplifed} into Equation~\ref{eqn:temp_ss} gives:
\begin{equation} \label{fixed_point_ori}
(1-a)T - bP_C = bV \kappa_{1}T^2e^{\frac{\kappa_{2}}{T}}
\end{equation}
If this equation has feasible solution(s), 
we can say that fixed points exist.
Since Equation~\ref{fixed_point_ori} is the \textit{focal point} of the subsequent analysis, 
we introduce the following change of variables 
to leave the exponential term alone and facilitate the subsequent analysis:
\begin{equation} \label{change_of_parameters}
\hspace{-2mm} \T = -\frac{\kappa_2}{T}, \hspace{3mm} 
\alpha = \frac{b}{a-1}\frac{P_C}{\kappa_2} > 0, \hspace{3mm}
\beta =  \frac{a-1}{b} \frac{1}{V \kappa_{1} \kappa_{2}} > 0 \hspace{1mm}
\end{equation}
With this change of variables, we rewrite Equation~\ref{fixed_point_ori} as:
\begin{equation} \label{eqn:fixed_point}
\beta \T(1-\alpha\T)=e^{-\T}
\end{equation}
where $\alpha >0, \beta >0$.
\textit{\uline{We will first derive the conditions on $\alpha$ and $\beta$ such 
that
Equation~\ref{eqn:fixed_point} has a solution.}}
Then, we will go back from the transformed domain to the original parameters. 
Finally, we will show how to compute the constraint on the maximum power consumption $P_C^*$ 
required to avoid thermal violations, given a maximum temperature 
constraint $T^*$.

\subsection{Necessary and Sufficient Conditions for the Existence 
of Fixed Point(s)} \label{sec:fixed_point_condition}

The domain of the auxiliary temperature is given by $\T \in (0,\infty)$, 
since $\T = -\kappa_2 /T$ where $\kappa_2 < 0$. 
Hence, the right-hand side of Equation~\ref{eqn:fixed_point} 
lies in the interval $(0,1)$. 
That is, $0 < \beta \T(1-\alpha\T) = e^{-\T} < 1$. 
Since this condition requires that $\T<1/\alpha$, 
we can take the logarithm of both sides while the equality holds, i.e.,
\[
\ln \beta + \ln \T + \ln(1-\alpha \T) = -\T
\]
Equation~\ref{eqn:fixed_point} has the same fixed points as the following equation: 
\begin{equation} \label{eqn:fixed_point_ln}
\mathcal{F}(\T) \triangleq \ln \beta + \ln \T + \ln(1-\alpha \T) + \T = 0
\end{equation}
%
The important properties of $\mathcal{F}(\T)$
employed in our analysis 
are summarized in the following lemma.
\begin{lemma} \label{lemma_F_T}
$\mathcal{F}(\T)$ given in Equation~\ref{eqn:fixed_point_ln} satisfies the following properties:
\begin{enumerate} [leftmargin=*]
  \item $\mathcal{F}(\T)$ is a concave function in the interval $\T \in 
  (0,1/\alpha)$.
  \item $\mathcal{F}(\T)$ has a unique maxima at $\T_m$, which is given by:
\begin{equation} \label{eqn:inflection_point}
\T_m = \frac{1}{2\alpha}-1+\sqrt{\frac{1}{4\alpha^2}+1}
\end{equation}
  \item $\mathcal{F}(\T)$ is an increasing function in the interval $(0,\T_m)$ and 
    a decreasing function in $(\T_m,1/\alpha)$.
\end{enumerate}
\end{lemma}

\begin{proof}
	The proof is provided in Appendix~\ref{proof_F_T}, 
	while an informal explanation is provided below to maintain a smooth flow.
\end{proof}

Figure~\ref{fixed_point_illustration} sketches $\mathcal{F}(\T)$ for $\T \in (0, 1/ \alpha)$.
As the first two properties of Lemma~\ref{lemma_F_T} state, 
$\mathcal{F}(\T)$ is concave and has a unique maxima at $\T=\T_m$. 
Furthermore, it is increasing in the interval $(0,\T_m)$ and decreasing in 
$(\T_m,1/\alpha)$, 
as shown in Figure~\ref{fixed_point_illustration}.
Hence, Equation~\ref{eqn:fixed_point_ln} ($\mathcal{F}(\T) = 0$)  
has \emph{two solutions} if and only if the maxima is non-negative, i.e., 
$\mathcal{F}(\T_m) \geq 0$. 
The fixed points coincide when $\mathcal{F}(\T_m) = 0$, 
and there are no fixed points for $\mathcal{F}(\T_m) < 0$. 
To be practical, these conditions should be expressed 
in terms of the design parameters $\alpha$ and $\beta$.  
This is achieved by Theorem~\ref{thm_fixed_point}, which summarizes our first major result.

\begin{figure}[b]
	\vspace{-3mm}
	\centering
	\includegraphics[width=0.65\linewidth]{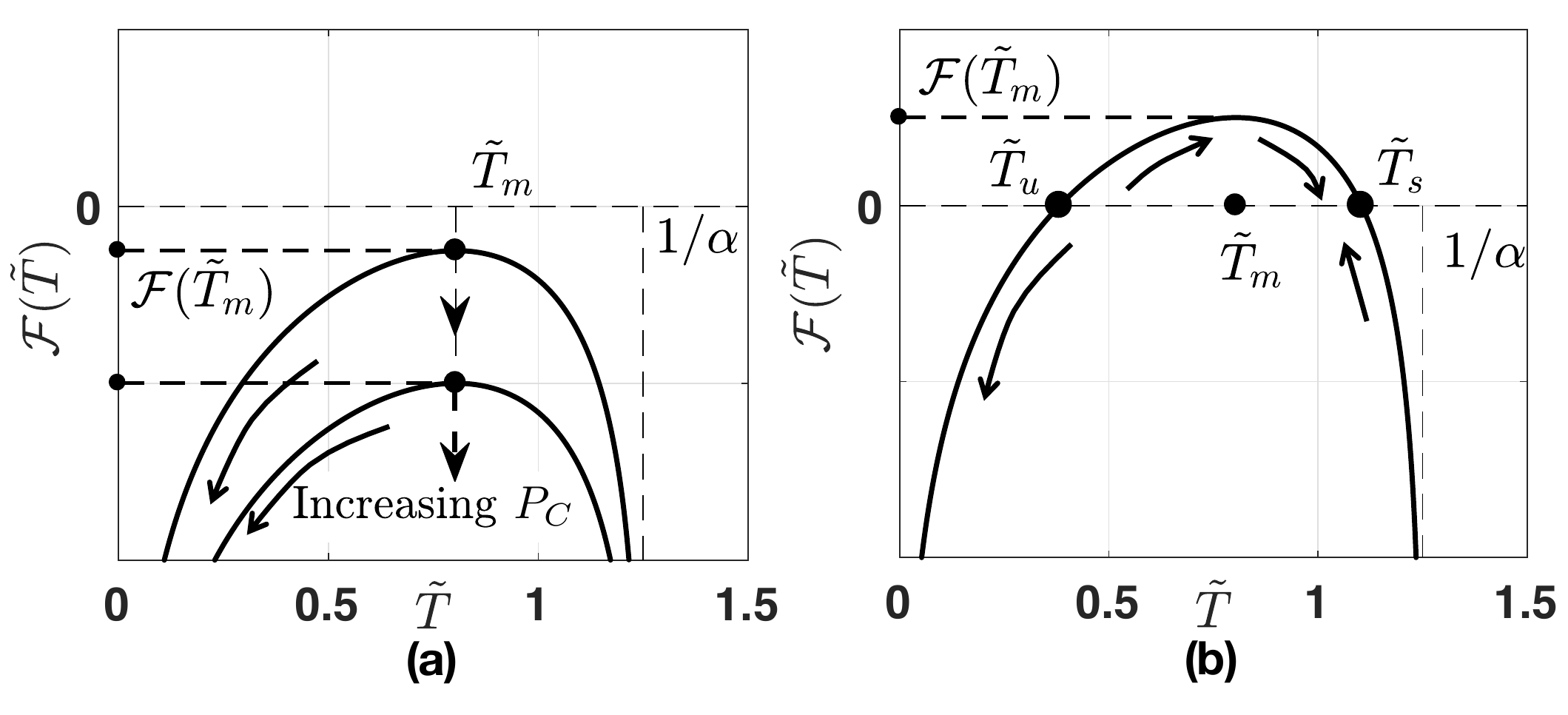}
	\vspace{-5mm}
	\caption{Illustration of $\mathcal{F}(\T)$ when there is no fixed point and 
	when there are two fixed points, respectively.}
	\label{fixed_point_illustration}
	\vspace{-3mm}
\end{figure}

\begin{theorem} \label{thm_fixed_point} 
The maxima of $\mathcal{F}(\T)$ is given by:
\vspace{-1mm}
\begin{equation} \label{eqn:F_Tm}
\mathcal{F}(\T_m) = \ln \beta - \ln\left(\frac{2}{\T_m}+1 \right)e^{-\T_m}
\end{equation}
Hence, Equation~\ref{eqn:fixed_point} has two fixed points if and only if $\beta\geq\left(\frac{2}{\T_m}+1\right)e^{-\T_m}$ where $\T_m$
depends only the parameter $\alpha$ and it is defined in 
Equation~\ref{eqn:inflection_point}.
Otherwise, it has no solution.
\end{theorem}
%
\begin{proof}
The proof is provided in Appendix~\ref{existence_proof}.
\end{proof}

At runtime, we first compute $\T_m$ using Equation~\ref{eqn:inflection_point}. 
Then, we check the condition given in Theorem~\ref{thm_fixed_point}. 
If it is not satisfied, we conclude that there will be a thermal runaway. 
This knowledge can be used to throttle the cores aggressively or enter an 
emergency state. 
Otherwise, we proceed to compute the maximum allowed power consumption 
that will avoid thermal violations.

\subsection{Stability of the Fixed Points}  \label{sec:stability}
The stability of the fixed points is determined 
by the behavior of $\T$ as function of $\mathcal{F}(\T)$. 
To provide a smooth flow, we summarize this behavior using the following lemma.
\begin{lemma} \label{lem_signFT}
The value of $~\T$ in the temperature iteration increases when 
$\mathcal{F}(\T)<0$, 
and decreases when $\mathcal{F}(\T)>0$.
\end{lemma}
\begin{proof}
The proof is provided in Appendix~\ref{proof_lem_sign_FT}.
\end{proof}
 
This lemma allows us to determine the stability characteristics of fixed points by inspecting the sign of the function $\mathcal{F}(\T)$. The following theorem summarizes the stability results using this lemma, which is also illustrated with the arrows in Figure~\ref{fixed_point_illustration}(a).

\begin{theorem} \label{thm_stability}
The stability of the fixed points is as follows.
\begin{enumerate} [leftmargin=*]
  \item When Equation~\ref{eqn:fixed_point_ln} has no solution,
the temperature iteration diverges, i.e., $\T \rightarrow 0$ $(T \rightarrow 
\infty)$, 
as illustrated in Figure~\ref{fixed_point_illustration}.
Hence, there is a thermal runaway.

  \item When there are two fixed points, $\T_u\in(0,\T_m)$ is unstable and $\T_s\in(\T_m,\frac{1}{\alpha})$ is stable.
      
  \item In the latter case, any temperature iteration starting $(0,\T_u)$  
  diverges, i.e., $\T \rightarrow 0$ and $T \rightarrow \infty$ leading to a 
  thermal runaway.
However, any temperature iteration starting in $(\T_u,\frac{1}{\alpha})$ 
converges to the stable  fixed point $\T_s$.
\end{enumerate}
\end{theorem}

\begin{proof}
The proof is provided in Appendix~\ref{stability_proof}, 
while an informal explanation is provided below.
\end{proof}

Theorem~\ref{thm_stability} states that the temperature 
proceeds along the arrows shown in Figure~\ref{fixed_point_illustration} during 
fixed point iterations.
When $\mathcal{F}(\T_m) < 0$, i.e., no fixed point exists, 
$\T$ decreases at each temperature iteration no matter where it starts. 
Therefore, there is a thermal runaway as illustrated in Figure~\ref{fixed_point_illustration}(a). 
When $\mathcal{F}(\T_m) > 0$, there are two fixed points denoted by $\T_u$ and $\T_s$. 
Any iteration starting in the interval $(0,\T_u)$ diverges to $\T \rightarrow 0$, 
since $\mathcal{F}(\T) < 0$ in that interval.
Conversely, any iteration starting in the interval $(\T_u, \T_s)$ will converge 
to $\T_s$, since $\mathcal{F}(\T) > 0$. 
That is, any iteration starting in the interval $(\T_u,\frac{1}{\alpha})$ 
will also converge to $\T_s$, as illustrated in 
Figure~\ref{fixed_point_illustration}(b). 
Therefore, we conclude that $\T_u\in(0,\T_m)$ is unstable,  
while $\T_s\in(\T_m,\frac{1}{\alpha})$ is stable.

In summary, we derived the conditions for the existence of fixed points and 
their stability regions in terms of the auxiliary temperature $\T$, $\alpha$ and $\beta$. 
Next, we will derive the constraints on the dynamic power consumption required for thermally safe operation.



\subsection{From Temperature to Power Constraint} \label{sec:computation}
Suppose that the thermally safe temperature is given by $T^*$.
We need to work backwards starting with this constraint to find the maximum allowable $P_C^*$ (i.e., the sum of dynamic and gate leakage power consumption).
To this end, we first show the existence of $P_C^*$, 
and then, we provide a \emph{compact formula} to compute $P_C^*$ in terms of 
$T^*$ and the system parameters.

\vspace{2pt}
\noindent \textbf{Existence of $P_C^*$:}
We know that the system converges to a stable fixed point 
when there is no dynamic activity ($P_C \rightarrow 0$).
This means that the necessary and sufficient condition given in 
Theorem~\ref{thm_fixed_point} is satisfied, i.e., $\beta\geq\left(\frac{2}{\T_m}+1\right)e^{-\T_m}$.
Now, consider the other extreme where $P_C$ grows indefinitely due to heavy 
dynamic activity.
Equation~\ref{change_of_parameters} shows that $\alpha$ will also grow, since it increases linearly with $P_C$.
Growing $\alpha$, in turn, implies that $\T_m \rightarrow 0$ according to Equation~\ref{eqn:inflection_point}.
Hence, as the dynamic activity ($P_C$) increases, the term $\left(\frac{2}{\T_m}+1\right)e^{-\T_m}$ \emph{increases monotonically}.
This suggests that there exists a $\T_{m}$ where $\beta=\left(\frac{2}{\T_{m}}+1\right)e^{-\T_{m}}$ holds.
At the same time, this corresponds to the maximum power $P(\T_{m})$ with 
two 
fixed points.
Furthermore, Equation~\ref{eqn:F_Tm} shows that $\mathcal{F}(\T_m)$ decreases monotonically with decreasing $\T_m$,
as illustrated by Figure~\ref{fixed_point_illustration}(a).
Therefore, $\mathcal{F}(\T)$ decreases monotonically as $P_C$ 
increases, and crosses the $x$-axis zero at $P_C^*$.
Due to the monotonic behavior of stable fixed points as a function of $P_C$,
we conclude that there exists a maximum allowable $P_C^*$ whose thermal fixed 
point does not exceed $T^*$.

\vspace{2pt}
\noindent \textbf{$P_C^*$ for a given $T^*$ is computed as follows at runtime:} 

\begin{framed}
\begin{enumerate} [leftmargin=*]
	\vspace{-1mm}
  \item Compute the auxiliary temperature that corresponds to $T^*$
using Equation~\ref{change_of_parameters} \\as  $\T^*= -\kappa_2 / T^*$
  \item Given $\T^*$, find the corresponding $\alpha^*$ using 
  Equation~\ref{eqn:fixed_point}: 
  \vspace{1mm}
  $\alpha^*=\frac{1}{\T^*} \left( 1-\frac{e^{-\T^*}}{\beta \T^*} \right) $ 
  \item Finally, substitute $\alpha^*$ into Equation~\ref{change_of_parameters} 
  to find $P_C^*$: 
  \vspace{1mm}
  $P_C^*=\frac{a-1}{b}\alpha^*\kappa_2 $ 
 \vspace{-1mm}
\end{enumerate}
\end{framed}
%
In summary, we conclude that any power value such that $P_C<P_C^*$ 
has a thermal stable fixed point less than $T^*$.

\subsection{Time to Reach the Stable Fixed Point}\label{sec:time_to_fp}
The existence and specific value of the fixed point can enable 
DTPM algorithms to determine whether the power/temperature trajectory 
moves towards a dangerous operating point. 
In addition to this, the expected time to reach to that point reveals 
how soon the dangerous zone will be reached. 
Therefore, DTPM algorithms can utilize the timing information 
to determine if there is any imminent possibility of a thermal violation. 
Moreover, this estimate can also be used 
to decide how long the current power consumption can be sustained 
without violating the thermal limit. 
To obtain a computationally efficient estimation, we employ the following exponential model:
\begin{equation}\label{eqn:exponential_model}
T[kT_s] = T_{init} + (T_{fix} - T_{init})(1 - e^{-\frac{kT_s}{\tau}})
\end{equation}
where $T_{init}$ is the initial temperature, $T_{fix}$ is the stable fixed 
point and $\tau$ is the time constant. 
We use the discrete time stamp $kT_s$, since the temperature is sampled with a 
period of $T_s$. 
Two temperature readings 
separated by a delay $DT_s$, i.e., $T[kT_s]$ and $T[(k-D)T_s]$,
can be used to estimate the time constant $\tau$ as:
\begin{equation}\label{eqn:tau_model}
\tau = \frac{t[kT_s] - t[(k-D)T_s]}{\ln(\frac{T[(k-D)T_s] - T_{fix}}{T[kT_s] - T_{fix}})}
\end{equation}
In our experiments, $D=10$ and the time to fixed point 
computation, 
like the other computations in the proposed technique, 
are repeated with $T_s =100$~ms.
\subsection{Solution for the MIMO Case} \label{sec:mimo}
The SISO model given in Equation~\ref{eqn:temp_ss} can be extended 
to explicitly show the impact of each power source on the thermal hotspot 
of interest $T_i$, as follows:
\begin{align*} 
T_i = {} & [a_{i1}, a_{i2}, \ldots, a_{iN}]\mathbf{T[k]} + [b_{i1}, b_{i2}, 
\ldots, b_{iM}] 
\mathbf{P[k]} \\
= {} & [a_{i1}, \ldots, a_{iN}] [T_1, \ldots, T_N]^\intercal + [b_{i1}, 
\ldots, b_{iM]} [P_1, \ldots, P_M]^\intercal \nonumber
\end{align*}
where $P_i$ terms are defined in Equation~\ref{eqn:total_power}.
In order to find the fixed point temperatures, we need to solve the set of 
equations given by:
\begin{equation}\label{eqn:temp_miso}
\mathbf{f(T)}=\mathbf{f}(T_1,\ldots,T_N)=\left[\begin{array}{c}
f_1(T_1,\ldots,T_N) \\
\vdots \\
f_N(T_1,\ldots,T_N) 
\end{array}\right]=0
\end{equation} where $
f_i(T_1,\ldots,T_N)=[a_{i1}, \ldots, a_{iN}] [T_1, \ldots, 
T_N]^\intercal$ \\
\hspace*{7em}$+ [b_{i1}, \ldots, b_{iM}] [P_1, \ldots, 
P_M]^\intercal - T_i,~\forall~1\leq i \leq N$.\vspace{1.5mm}\\
This is a nonlinear equation due to the exponential terms in $P_i$. 
Generalizing our solution to obtain similar closed-form formulae for the MIMO 
model requires solving the nonlinear equations in Equation~\ref{eqn:temp_miso}. 
The standard approach to this problem 
is to find a good initial point, and use a root search algorithm for nonlinear equations. 
By following this approach, we use our SISO solution for each hotspot as the ``\textit{initial point}''. 
Then, we employ an iterative search to find the roots of nonlinear equations. 
More 
specifically, we solve Equation 17 via Newton's method~\cite{atkinson2008introduction} where our SISO 
solution is used as the initial point. 
Our experimental measurements show that the closed loop 
solutions 
presented in Sections~\ref{sec:fixed_point_condition}-\ref{sec:time_to_fp} 
are very close to the roots of the MIMO system. 
Hence, iterations converged to the roots of the MIMO system in less than five 
iterations in the worst case. 
Moreover, our simulations show that the region of convergence is quite large. 
This means that even if our initial points are not accurate, we will still converge to the solution. 
\section{Experimental Evaluation} \label{sec:experiments}


\subsection{Experimental Setup} \label{sec:experimental_setup}
The proposed fixed point prediction scheme is evaluated on the Odroid XU3 board
which employs Samsung Exynos 5422 SoC~\cite{ODROID_Platforms}.
Exynos 5422 is a single ISA big.LITTLE SoC consisting of 4 Cortex-A15~(big) 
cores, 4 Cortex-A7~(little) cores and a Mali GPU.
The Odroid XU3 board also comes with thermal sensors to measure the temperature 
of each big core and the GPU.
It also includes current sensors to measure the power consumption of the big 
cluster, little cluster, the memory and the GPU.
The proposed approach is invoked with the default frequency 
governors every 100~ms.
	That is, we periodically sample four power consumption and five temperature 
	values.
	The proposed technique takes 75.2~$\mu$s of this 100~ms period, as detailed 
	in Section~\ref {sec:implementation_overhead}.
We evaluated the proposed approach on commonly used benchmarks,
	including four from the MiBench suite~\cite{guthaus2001mibench},
	three from the PARSEC suite~\cite{Bieniaparsec}, one 
	SPEC~\cite{henning2006spec}, and a matrix multiplication kernel.
	These benchmarks are listed in the first column of 
	Table~\ref{table:summary}.

\subsection{Parameter Characterization} \label{sec:sys_id}
The stability and safety analysis presented in 
Section~\ref{sec:fixed_point_prediction} does not make any assumptions about 
specific parameter values. 
However, we need to characterize the system parameters, such as the leakage current and 
the state-space model parameters, in order to evaluate the accuracy of our 
analysis. 
As the first step, we characterized the leakage current parameters $I_g, \kappa_1, \kappa_2$ for the Odroid XU3 board 
using a methodology similar to those presented 
in~\cite{liu2007accurate,singla2015predictive}.
More precisely, we placed the target device in a furnace and collected data 
at five different temperatures ranging from 40$^\circ$C to 80$^\circ$C, 
while running a light workload to make sure that the temperature is not increased due to the dynamic activity. 
Since the dynamic power consumption remained constant, the measured power consumption difference between two experiments 
gives the difference in leakage power at two known temperatures. 
We used measurements at five different temperatures to find the unknowns 
$I_{g,i}, \kappa_{1,i}, \kappa_{2,i}$ for the big core cluster and the GPU as these are the dominant sources of leakage in our board.

\begin{figure}[t]
	\centering
	\includegraphics[width=0.59\linewidth]{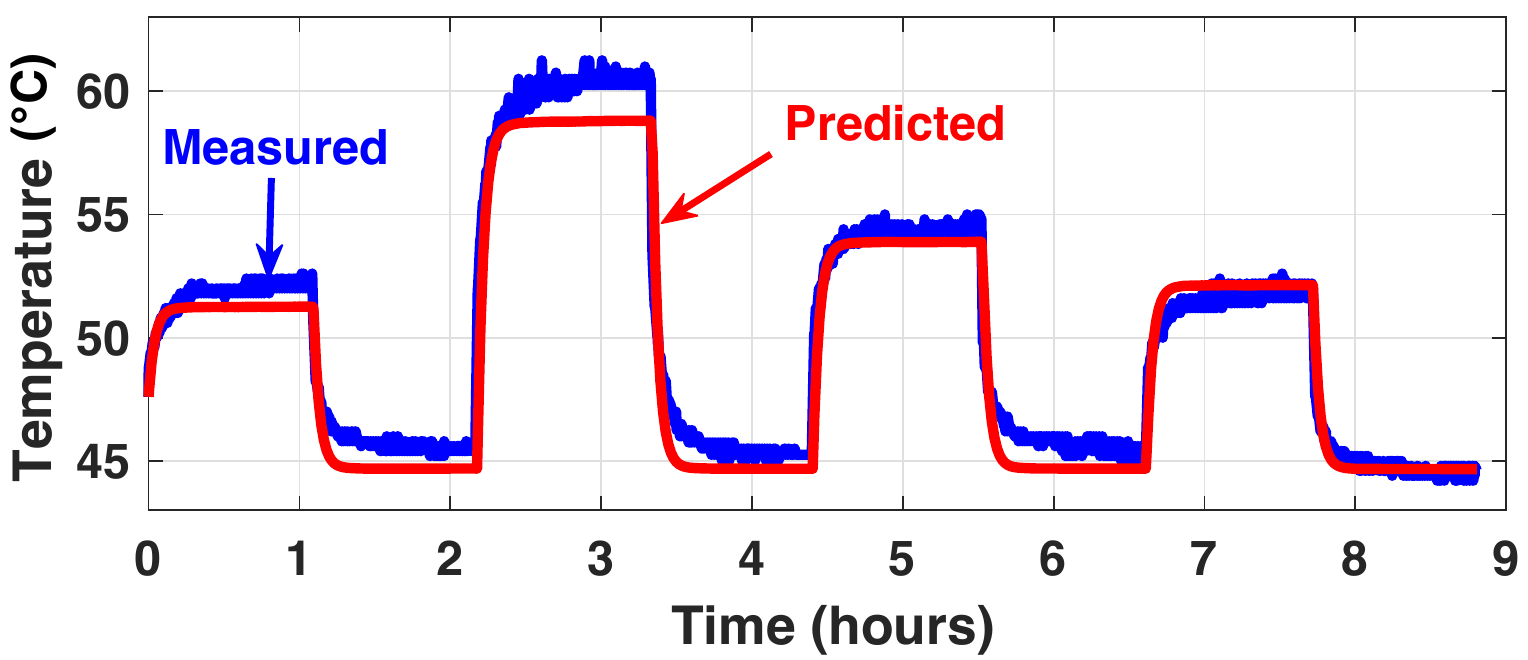}
	\vspace{-0.15in}
	\caption{Comparison of predicted temperature with the measured temperature.}
	\label{fig:iterations}
	\vspace{-0.07in}
\end{figure}

The second step is estimating the $A$ and $B$ matrices used in Equation~\ref{eqn:temp_model}.
In order to characterize the $A$ matrix, we ran the big cluster at a fixed frequency level until the temperature rises to a steady state value.
Once the steady state was reached, we turned off the big core and let the device run until the board cools down. 
We repeated this for multiple frequency levels as shown in Figure~\ref{fig:iterations}.
We also repeated the complete experiment multiple times and then averaged the 
results to account for any variations in the system behavior over time.
The data from this experiment helps us determine how the temperature at the current time step affects the temperature in the next time step.
In a typical run, different workloads lead to varying utilization of system 
resources, 
which lead to different operating frequencies. 
To determine the effect of each power source on the temperature, 
we excited all the frequency levels of the A7 cores, A15 cores and the GPU in a pseudo-random bit sequence~(PRBS),
and recorded the temperature.
This helps us to identify the effect of each power source on the temperature 
under varying conditions, 
i.e., characterizing the $B$ matrix.
This data is then used in conjunction with the data from the previous experiment to jointly identify $A$ and $B$.
The $A$ and $B$ matrices are used to predict the temperature given an initial temperature and power consumption at each time step. 
As we can see in Figure~\ref{fig:iterations}, the computed temperature closely 
follows the measured temperature.
In this work, we performed characterization on one Samsung Exynos 5422 chip. 
The parameters obtained for leakage current and the thermal model are used to 
evaluate the accuracy of the fixed point analysis on three different instances. 
In general, the chip manufacturers divide their products into multiple bins, 
such as low power and high performance parts, 
and sell them using different stock keeping units~(SKUs). 
Therefore, our analysis can be performed for each SKU of a given commercial 
chip.
%

\subsection{Validation of Fixed Point Approximation}\label{sec:fixed_point_eval}
\noindent\textbf{Evaluation at fixed frequency:} To evaluate 
the accuracy of the proposed fixed point analysis, we first performed 
experiments at 
various power levels ranging from 0.38 W to 1.01 W, 
while running a light load on the CPU. 
After running the system until a steady-state fixed point is attained, 
we recorded the average power consumption and final temperature. 
We also estimated the average dynamic power consumption for each experiment 
by subtracting the estimated leakage power from the measured total power. 
The first three columns in Table~\ref{table:summary} provide these values. 
Then, we used the estimated dynamic power consumption and the ambient temperature 
to analyze whether a stable fixed point exists or not. 
After confirming that the analysis result is correct, 
we computed the fixed point temperature using the proposed technique. 
The fourth column of Table~\ref{table:summary}, labeled as ``Comput. Fixed 
Point Temp.'', 
lists the results for each power consumption level. 
As summarized in Table~\ref{table:summary}, the fixed point prediction is within 1$^\circ$C of the empirical result for four power levels. 
The largest observed difference was 1.1$^\circ$C, which implies 2.0\% error 
with respect to the measured fixed point. 
This error is quite acceptable for our system, as the temperature sensors 
operate at an integer precision, which can introduce an error in the 
measurement.

\vspace{2mm}
\noindent\textbf{Evaluation on benchmarks:} 
We also evaluated our analysis technique on commonly used benchmarks 
which represent real-world applications. 
We ran these applications for minutes to capture the temperature dynamics, 
as summarized in the last column of Table~\ref{table:summary}.
During these experiments, we overwrote the safe temperature limit on the board 
such that the temperature could rise beyond 110$^\circ$C.
The results of the fixed point evaluation for the benchmarks are summarized 
in the lower part of Table~\ref{table:summary}. 
We observe that the fixed point prediction error increases slightly compared to the fixed frequency experiments. 
This is expected, since there are larger variations in the power consumption. 
However, the average prediction error is still only 
3.0$^\circ$C, 
and in the worst case, the error is 5.8$^\circ$C. 
Finally, the average prediction error across all experiments is 
2.3$^\circ$C.

\begin{table*}[t!]
	\footnotesize
	\centering
	\caption{Summary of results for fixed point prediction. 
		The foreground application is listed in the first column. 
		Since these experiments were performed under Android OS, 
		there were more than 100 background applications running at all times.}
	\label{table:summary}
	\begin{minipage}{\columnwidth}
		\begin{center}
			\vspace{-0.1in}
			\begin{tabular}{@{}lccccccc@{}}
				\toprule
				Benchmark &
				\begin{tabular}[c]{@{}c@{}}Avg. Total\\ Power (W)\end{tabular} 
				& 
				\begin{tabular}[c]{@{}c@{}}Avg. Dyn.\\ Power (W)\end{tabular} & 
				\begin{tabular}[c]{@{}c@{}}Empirical Fixed \\Point Temp.
					($^\circ$C)\end{tabular} & 
					\begin{tabular}[c]{@{}c@{}}Comput. 
					Fixed\\
					Point Temp. ($^\circ$C)\end{tabular} & 
					\begin{tabular}[c]{@{}c@{}} 
					Abs. Pred.\\ Error ($^\circ$C)\end{tabular} &  \hspace{-1mm}
				\begin{tabular}[c]{@{}c@{}}Perc. Pred.\\ Error (\%) 
				\end{tabular}
				& \hspace{-1mm}\begin{tabular}[c]{@{}c@{}}Runtime\\ 
					(s)\end{tabular} 
				\\ \toprule
				\begin{tabular}[l]{@{}l@{}}Idle @  1.3 GHz\end{tabular} 
				& 
				0.38                                                           
				& 
				0.31                                                            
				 & 
				51.8                                                            
				        
				&
				
				52.2                                                            
				        
				&
				0.4                                                             
				       
				& 0.8 & 3979
				\\ 
				\midrule
				\begin{tabular}[l]{@{}l@{}}Idle @  1.5 GHz\end{tabular} & 
				0.47                                                           
				& 
				0.38                                                            
				 & 
				54.4                                                            
				        
				&
				
				55.5                                                            
				        
				&
				1.1                                                             
				       
				& 2.0 & 4045 
				\\ 
				\midrule
				\begin{tabular}[l]{@{}l@{}}Idle @  1.8 GHz\end{tabular} & 
				0.70                                                           
				& 
				0.59                                                            
				 & 
				60.2                                                            
				        
				&
				
				60.1                                                            
				        
				&
				0.1                                                             
				       
				& 0.2 & 4171 
				\\ 
				\midrule
				\begin{tabular}[l]{@{}l@{}}Idle @  2.0 GHz\end{tabular} & 
				1.01                                                           
				& 
				0.87                                                            
				 & 
				66.0                                                            
				        
				&
				
				66.6                                                            
				        
				&
				0.6                                                             
				       
				& 0.9 & 3413 
				\\ 
				\toprule
				%
				%
				%
				%
				%
				%
				%
				%
				%
				%
				%
				Vortex & 
				1.73                                                   
				& 
				1.55                                                   
				& 
				80.0                                                   
				
				&
				
				81.4                                                   
				
				&
				1.4                                                   
				
				& 1.7 & 1989 
				\\ 
				\midrule
				\begin{tabular}[l]{@{}l@{}}Matrix Mult. \end{tabular} & %
				1.84                                                           
				& 
				1.65                                                            
				 & 
				83.0                                                            
				        
				&
				
				83.8                                                            
				        
				&
				0.8                                                             
				       
				& 1.0 & 521 
				\\ 
				\midrule
				CRC32 & 
				2.04                                                           
				& 
				1.83                                                            
				 & 
				85.0                                                            
				        
				&
				
				88.5                                                            
				        
				&
				3.5                                                             
				       
				& 4.1 & 907
				\\ 
				\midrule
				Patricia & 
				2.20                                                           
				& 
				1.97                                                            
				 & 
				89.0                                                            
				        
				&
				
				91.8                                                            
				        
				&
				2.8                                                             
				       
				& 3.1 & 900
				\\ 
				\midrule
				Blackscholes & 
				2.42                                                   
				& 
				2.17                                                   
				& 
				94.0                                                   
				
				&
				
				96.6                                                   
				
				&
				2.6                                                    
				
				& 2.7 & 785
				\\ 
				\midrule
				Streamcluster & 
				2.48                                                   
				& 
				2.22                                                            
				 & 
				94.0

				&
				
				97.9

				&
				3.9                                                    
				
				& 4.1 & 614
				\\
				\midrule
				Basicmath & 
				2.49                                                           
				& 
				2.24                                                            
				 & 
				93.0                                                            
				        
				&
				
				98.3                                                            
				        
				&
				5.3                                                             
				      
				& 5.7 &  492 
				\\ 
				\midrule
				Fluidanimate &  
				2.50                                                           
				& 
				2.26                                                            
				 & 
				93.0                                                            
				        
				&
				
				98.8                                                            
				        
				&
				5.8                                                             
				       
				& 6.2  & 
				550                                                         \\
				\midrule
				FFT & 
				2.71                                                           
				& 
				2.41                                                            
				 & 
				101.0

				&
				
				102.5

				&
				1.5                                                             
				       
				& 1.5 & 
				729                                                           \\
				\midrule
				Streamcluster \hspace{-0.5mm}+ \hspace{-0.5mm}CRC32  \hspace{-7mm} &
				3.14                                                  
				& 
				2.79                                                  
				& 
				109.0

				&
				
				111.5

				&
				2.5                                                   
				
				& 2.3 & 
				1132                                                   
				
				\\ 
				\bottomrule
			\end{tabular}
		\end{center}
	\end{minipage}
\end{table*}

\subsection{Power-Temperature Trajectory}
We also performed experiments to compare the
measured power-temperature trajectory 
against the fixed point predicted by our analysis. 
Trajectories for different initial power consumption and temperature 
are shown in Figure~\ref{fig:fp_compare}.
The initial points are denoted by black $\circ$ markers and simulated trajectories are plotted using dashed black lines. 
For each of the initial points, the dynamic power starts with a small value and then rises to a steady value depending on the workload.
For example, consider the trajectory that starts from 
the initial point (0.5~W, 50$^\circ$C). 
The dynamic power increases until about 0.87~W, 
and then it remains steady at 0.87~W. 
The arrows show how each trajectory converges to the 
stable fixed point at (1.02~W, 66.6$^\circ$C) 
shown by a larger black $\bigcirc$. 
These trajectories are obtained by solving the set of nonlinear equations given 
in Equation~\ref{eqn:nonlinear_model} iteratively. 
Given the same assumptions, the proposed approach 
finds the stable fixed point as (1.02~W, 66.6$^\circ$C), 
as denoted by the red $\blacktriangle$ marker on the same figure. 
Hence, the prediction exactly matches the simulated trajectory. 
Furthermore, we performed one more set of experiments on the board by imposing a dynamic power consumption 
close to the value used in the simulation. 
The trajectory measured during this experiment is plotted using solid blue line in Figure~\ref{fig:fp_compare}.
This trajectory almost overlaps with the simulation and converges to (1.01~W, 
66.0$^\circ$C). 
As a result, the difference between the empirical and theoretical fixed points is only 
0.6$^\circ$C, which agrees with the results
in Table~\ref{table:summary}.

\begin{figure*}[t!]
	\begin{minipage}{1\textwidth}
		\hspace{0.8in}
		\includegraphics[width=0.7\textwidth]{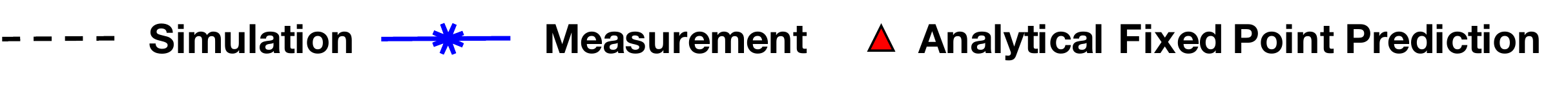}
		\vspace{-2mm}
	\end{minipage}
	\begin{minipage}{0.45\textwidth}
		\centering
		\includegraphics[width=1\linewidth]{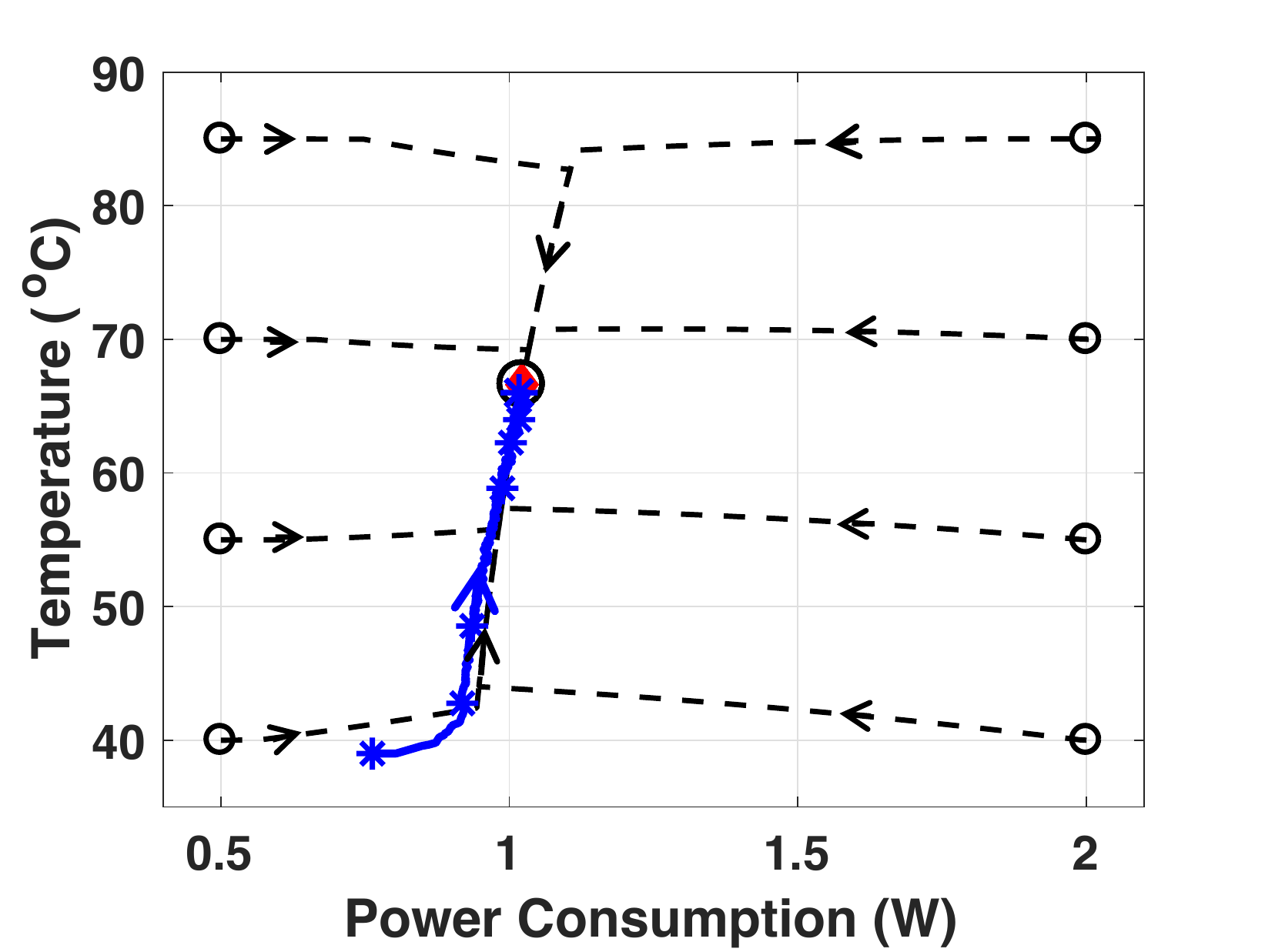}
		\vspace{-3mm}
		\subcaption{Analytical, experimental and simulation 
			results when the fixed point is \textit{lower} than the thermally 
			safe 
			temperature.}
		\label{fig:fp_compare}
	\end{minipage}
	\hfill
	\begin{minipage}{0.45\textwidth}
		\centering
		\includegraphics[width=1\linewidth]{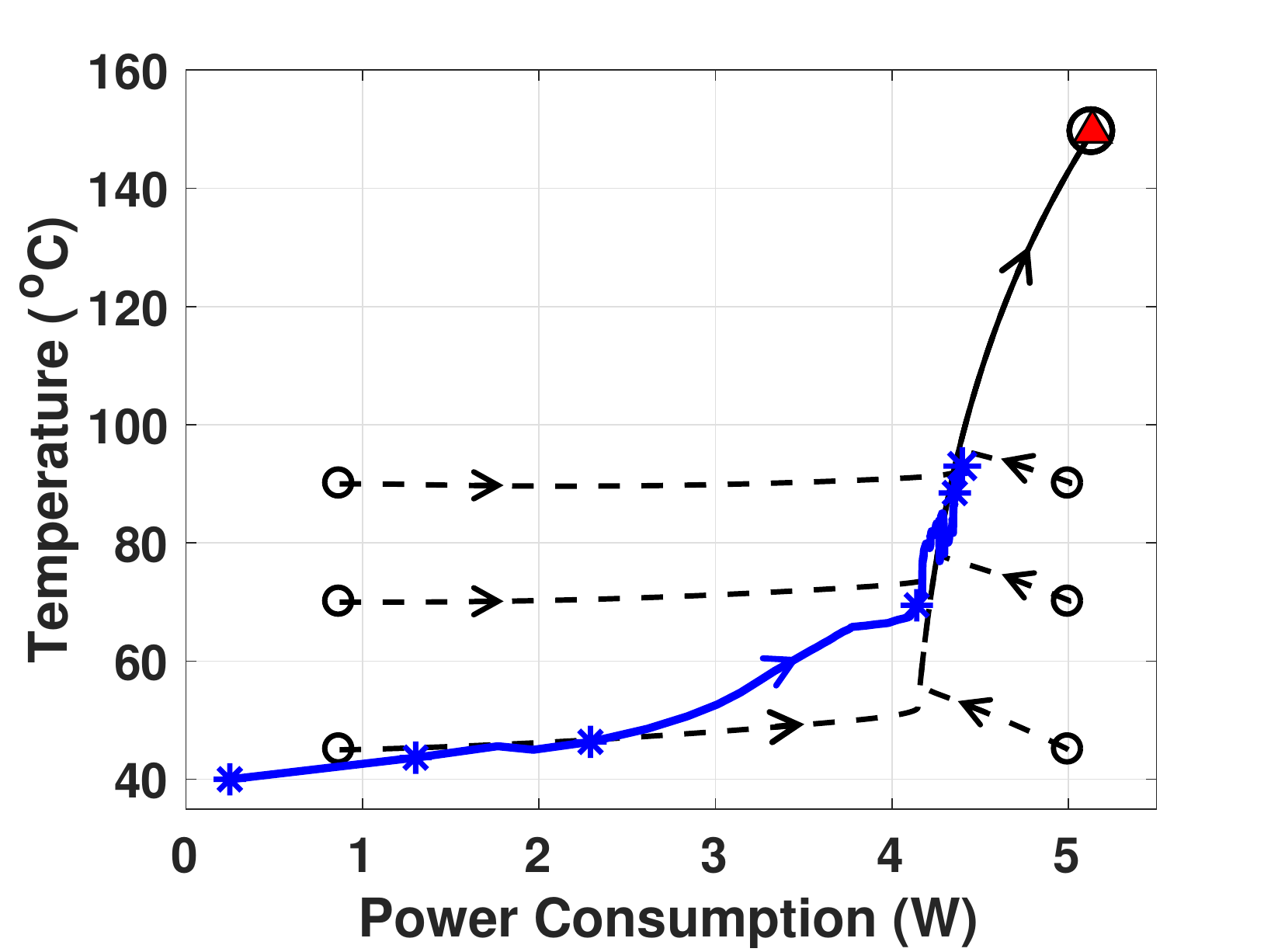}
		\vspace{-3mm}
		\subcaption{Analytical, experimental and simulation 
			results when the fixed point is \textit{higher} than the thermally 
			safe 
			temperature.}
		\label{fig:unstable_simulation}
	\end{minipage}
	\vspace{-0.1in}
	
	\caption{Power-temperature trajectories with multiple initial power and 
		temperature values. Black $\circ$ markers are the initial points and 
		the 
		dotted black line shows the trajectory followed by power and 
		temperature. 
		Black $\bigcirc$ and red $\blacktriangle$ markers show the simulated 
		and 
		\textit{computed} fixed point respectively. 
		We observe that the simulation converges to the computed fixed point 
		for 
		each initial point. The trajectory of the real experiment~(shown using 
		blue 
		lines) follows the simulated trajectory closely.}
	\vspace{-0.1in}
\end{figure*}

In order to illustrate the case where the temperature converges to a value beyond the thermal limit of our board, 
we performed one more set of experiment and simulation, as shown in Figure~\ref{fig:unstable_simulation}. 
We disabled the thermal throttling to let the temperature 
rise beyond 110$^\circ$C.
Similar to Figure~\ref{fig:fp_compare}, the dynamic power starts with a small value until it increases to 4.23~W.
The dashed black lines show the simulated trajectory followed by the power and temperature for each initial point.
We note that the simulation for each initial point converges 
to the fixed point at (5.14~W, 149.5$^\circ$C)
marked with a larger black $\bigcirc$.
With these conditions, the proposed approach finds 
the fixed point as (5.14~W, 149.6$^\circ$C), 
as denoted by the red $\blacktriangle$ marker.
The difference between the analytical solution and 
simulated trajectory is 0.1$^\circ$C, hence, they almost overlap in 
Figure~\ref{fig:unstable_simulation}.
Moreover, we performed one instance of the experiment 
on the board by choosing the initial point as (0.4~W, 40$^\circ$C). 
We also forced the dynamic power 
consumption close to 4.20~W, i.e., the value used in the simulation.
The solid blue line in Figure~\ref{fig:unstable_simulation} shows that 
the actual trajectory closely follows the simulated trajectory 
until the temperature reaches 93$^\circ$C.
At that point, we stopped the experiment to avoid damage to our board.

\noindent\textbf{Power-Temperature with Thermal Throttling:} 
To further evaluate the validity of our fixed point prediction, 
we performed two sets of experiments: one with and 
the other without throttling. 
We used the FFT benchmark from the MiBench suite, 
as it exhibits a representative behavior and 
leads to higher temperature fixed point than the other applications.
The proposed fixed point calculation 
is performed every 100~ms in the Linux kernel.
This allows us to analyze how the fixed point prediction 
evolves over time. 
Red $\blacktriangle$ markers in Figure~\ref{fig:fft_fixed_point_eval} 
show that the fixed point predictions vary from about 
95.0$^\circ$C to 105.0$^\circ$C.
This variation is expected, since the dynamic 
power changes during the execution of the application. 
However, our analysis results still match closely with the 
measured trajectory which approaches the predicted fixed point. 


We repeated the previous experiment, this time by incorporating 
a thermal throttling policy.  
The power consumption starts from a small value of 
0.48~W and then increases to about 2.40~W 
after the benchmark starts execution.
When the power consumption is 2.40~W, 
the fixed point is predicted as 99.0$^\circ$C. 
As the power consumption steadily increases 
due to the increased leakage, 
the fixed point prediction increases to 102.0$^\circ$C, 
as shown by the region A in Figure~\ref{fig:fft_throttle_experiment}. 
The solid blue line shows that the measured trajectory 
advances towards the predicted fixed point, 
as in the previous experiment. 
This time, the DTPM policy is triggered to throttle the frequency 
as soon as the temperature reaches 85.0$^\circ$C. 
Throttling starts reducing the power consumption, 
which in turn slows down the temperature increase. 
Figure~\ref{fig:fft_throttle_experiment} shows that 
the reduction in the power consumption is reflected in our 
fixed point prediction. 
More precisely, the proposed technique updates the fixed point prediction 
as~(2.07~W, 87.7$^\circ$C).
At the same time, the measured power-temperature 
trajectory changes its course. 
We observe that it starts converging to (2.06~W, 
86.0$^\circ$C), 
which matches very well with our prediction. 
This experiment illustrates that our analysis can adapt to changes in 
the dynamic power consumption and predict the fixed point accurately.

\begin{figure*}[t!]
	\begin{minipage}{1\textwidth}
	\hspace{0.8in}
	\includegraphics[width=0.7\textwidth]{figures/legend.pdf}
	\vspace{-2mm}
\end{minipage}
	\begin{minipage}[b]{0.45\textwidth}
		\centering
		\includegraphics[width=\linewidth]{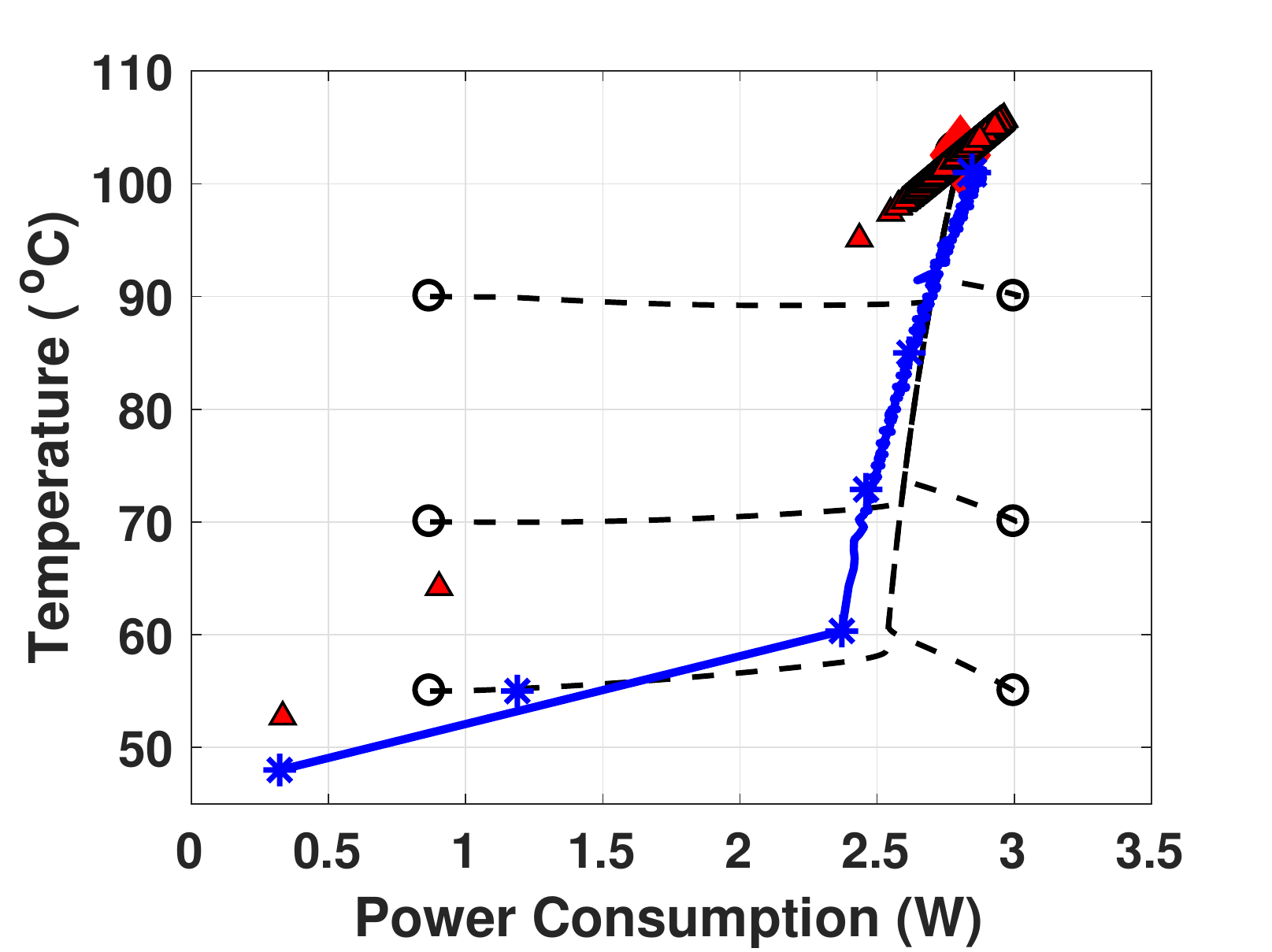}
		\caption{Evaluation of the power-temperature 
			trajectory 
			and fixed point evaluated every 100 ms
			for the FFT benchmark.}
		\label{fig:fft_fixed_point_eval}
	\end{minipage}
	\hfill
	\begin{minipage}[b]{0.45\textwidth}
		\centering
		\includegraphics[width=\linewidth]{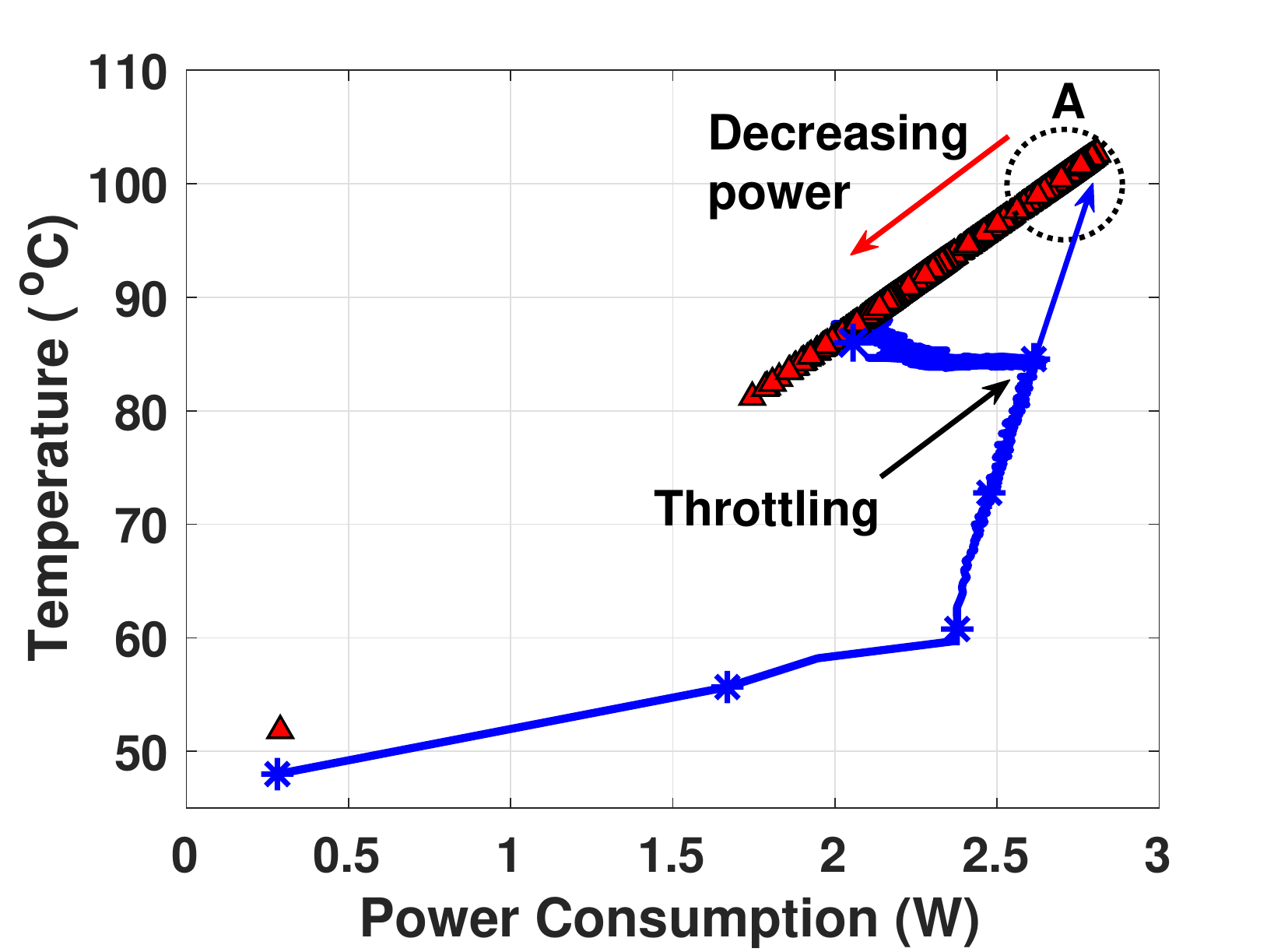}
		\caption{Using the proposed analysis for thermal 
			throttling, 
			and its impact on power-temperature trajectory.}
		\label{fig:fft_throttle_experiment}
	\end{minipage}
	\vspace{-5mm}
	
\end{figure*}
\subsection{Power Constraint from Temperature}
We also evaluated the change in the maximum power constraint 
$P_C^*$ when the 
temperature constraint $T^*$ is varied.
In particular, we swept the value of $T^*$ from $50^\circ$C to $105^\circ$C 
and calculated the value of $P_C^*$. 
The black line in Figure~\ref{fig:tstar_vs_pc} shows the maximum 
the power constraint $P_C^*$ found using the analytical approach 
outlined in Section~\ref{sec:computation}. 
To validate the analysis results, we set the temperature constraint $T^*$ 
as the theoretical fixed point. 
Then, we empirically found the power constraint 
$P_C^*$ on the target board. 
The red $\vartriangle$ markers in Figure~\ref{fig:tstar_vs_pc} 
show that the measured results are indeed on the 
trend found by the proposed technique. 
This result can be easily used as a guideline to decide 
the maximum power level that a chip can be operated based on a given 
temperature constraint.

\begin{figure*}[t!]
	\begin{minipage}[t]{0.45\textwidth}
		\centering
		\includegraphics[width=0.9\linewidth]{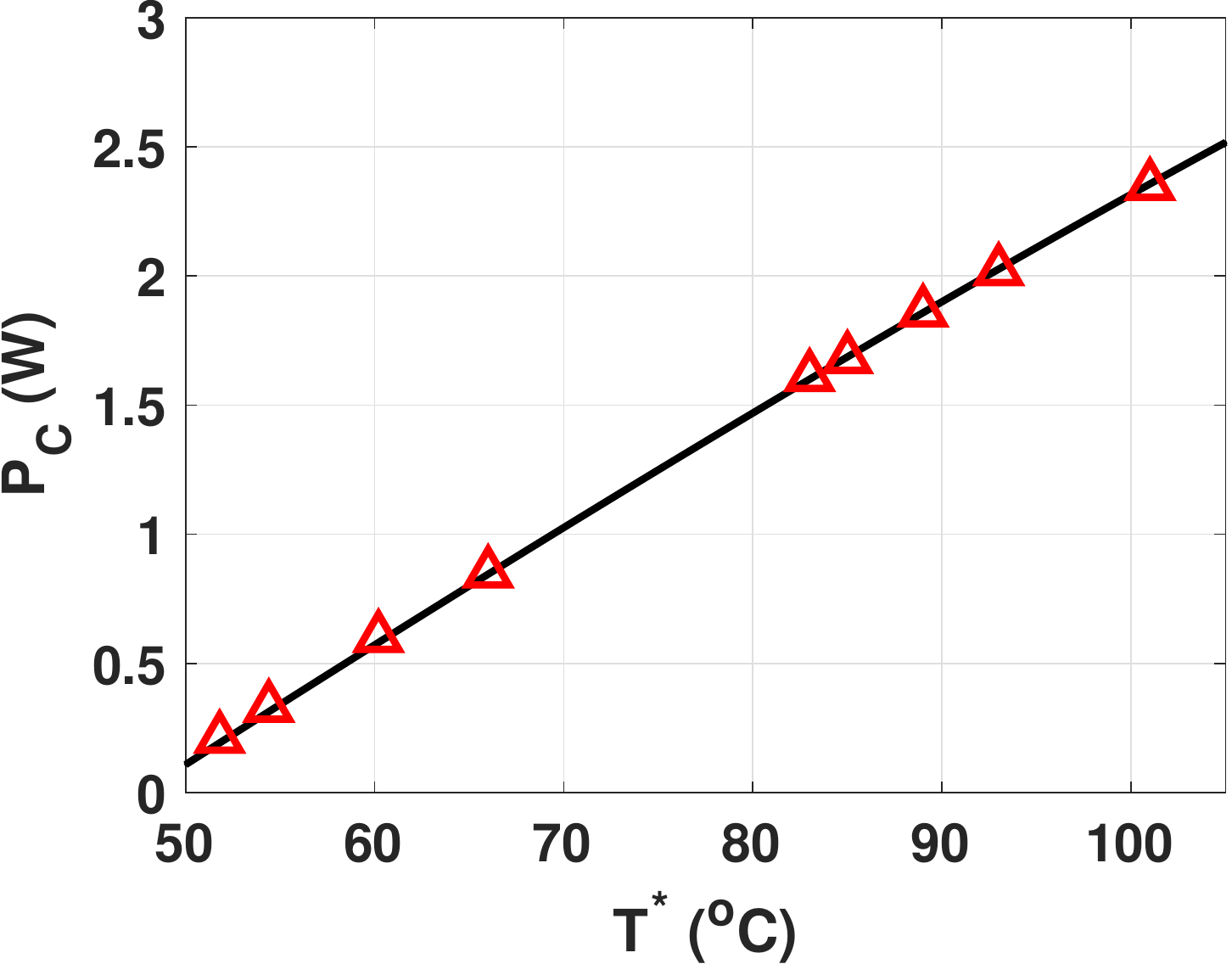}
		\caption{Variation of the maximum power constraint $P_C^*$ for 
			different 
			values of $T^*$.
		}
		\label{fig:tstar_vs_pc}
	\end{minipage}
	\hfill
	\begin{minipage}[t]{0.45\textwidth}
		\centering
		\includegraphics[width=0.95\linewidth]{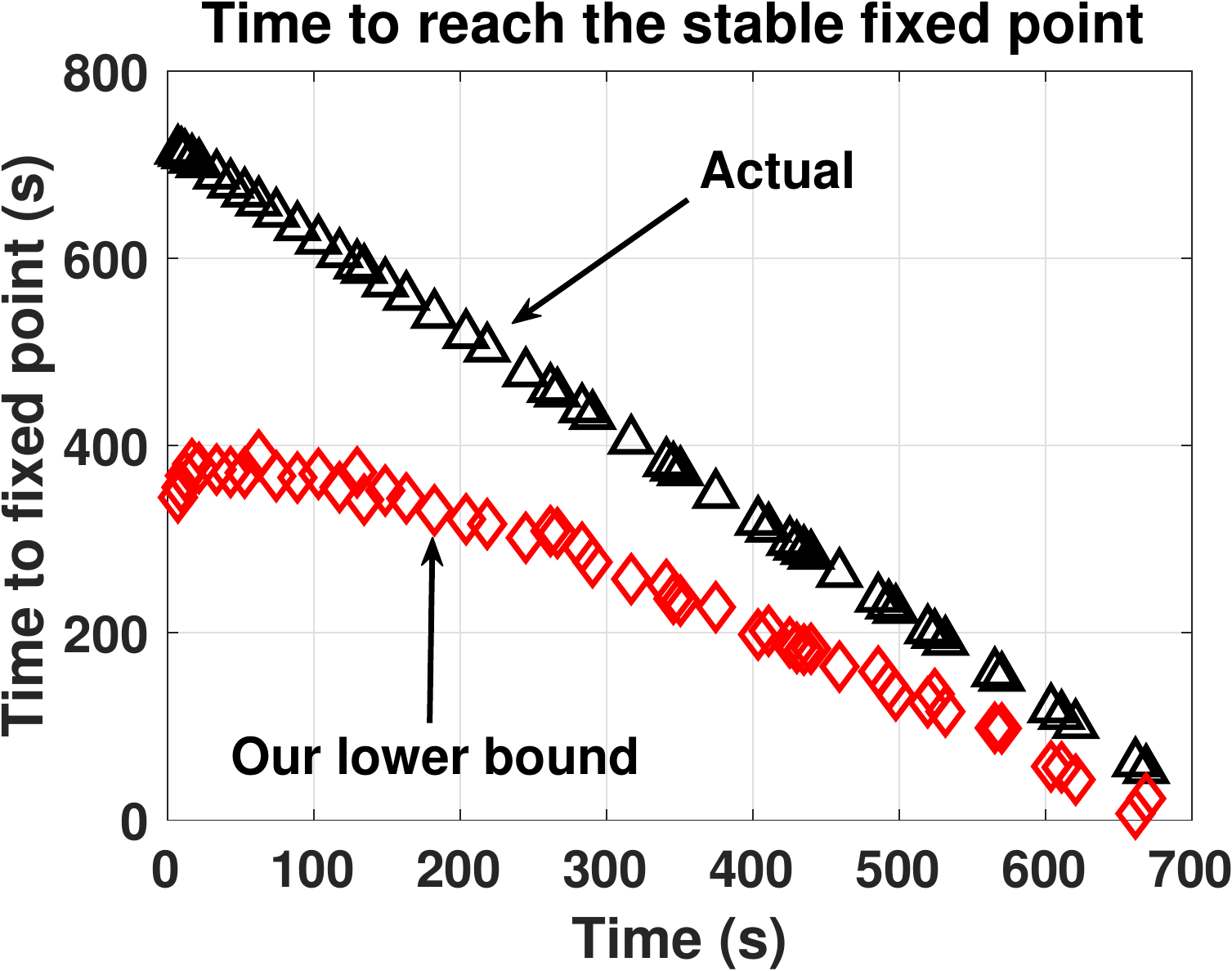}
		\caption{Comparison of the actual time taken to reach fixed point and 
			the 
			predicted time to reach fixed point.}
		\label{fig:time_to_fixed_point}
	\end{minipage}
		\vspace{-0.15in}
\end{figure*}

\subsection{Evaluation of the Time to Reach Fixed Point}

We used Equation~\ref{eqn:tau_model} to estimate the 
time at which the temperature will reach the fixed point for the FFT benchmark.
Then, we compared this estimate to the actual time 
to reach the fixed point. 
Figure~\ref{fig:time_to_fixed_point} shows that our estimate provides a lower bound 
for the time to reach the fixed point.
A lower bound is useful, since it can be used safely to avoid thermal violations.
We also see that the estimation improves in accuracy 
as the benchmark continues to run.
In summary, this estimation 
can be used by DTPM policies to decide how long the current 
power consumption level can be 
sustained without violating the thermal limit.



\subsection{Implementation Overhead}\label{sec:implementation_overhead}
Our theoretical analysis and proofs enable us to derive 
	analytical solutions that can be implemented with negligible overhead. 
	To quantify this overhead, we implemented the proposed solutions
	on Android 4.4.4 / Linux 3.10.9 kernel user space, 
	and measured the overhead on the Odroid XU3 mobile platform.
Our implementations are invoked periodically with the default 
frequency governors, i.e., every 100 ms. 
We observed that reading the sensors takes 
	13.8~$\mu$s, while computing the fixed point estimate for the SISO case
			takes 6.8~$\mu$s.
	We can achieve this low overhead since both $1/\alpha$ and
	$\T_m$ have closed form solutions. Once we have the fixed point estimate, 
	the Newton's method to solve the MIMO case takes about 50~$\mu$s.
Similarly, it takes 1.1 $\mu$s to compute the maximum allowable power 
consumption $P_C^*$ 
given a temperature constraint. 
This small overhead is enabled by the three closed form 
equations framed 
	at the end of Section~\ref{sec:implementation_overhead}.
Finally, computing the time to reach the stable fixed point 
takes 3.5~$\mu$s. The combined overhead of all three computations is about 
75.2~$\mu$s out of 100 ms, i.e., 0.075\%. When the 
implementation is moved to the kernel, the execution time of the functions 
reduces by about 30\%.
In contrast, an iterative approach cannot determine the existence and 
stability of fixed points.
Furthermore, it can predict the temperature given the power consumption, but it 
cannot compute the maximum allowable power consumption $P_C^*$.
Finally, temperature prediction over an interval of 1000~s alone takes about 
550 $\mu$s.
The iterative approach also has an average error of more than 10$^\circ$C, 
which is higher than that of our approach.

\section{Conclusion} \label{sec:conclusion}
This paper presents a theoretical analysis of the stability 
of the power consumption and temperature dynamics. 
First, we show that the power-temperature dynamics have either no fixed point 
or two fixed points, as a function of the system parameters and the dynamic 
power consumption.
When there are two fixed points, we prove that 
one of the fixed points is stable, while the second one is unstable.
We also determine the region of convergence, which is important for safe 
thermal 
operation.
Third, we derive an analytical formula to compute the maximum dynamic power 
consumption that guarantees a thermally safe operation.
Experiments and simulation results show that our analysis can be used 
to predict the fixed point within 0.1$^\circ$C to 5.8$^\circ$C 
accuracy with only 0.075 ms computational overhead.
Hence, the proposed approach can be used to take 
\emph{proactive} DTPM decisions, 
and detect security threats which force the system to operate beyond the 
thermal limit.

%
%
\bibliographystyle{ACM-Reference-Format}
\bibliography{references/embedded_refs}
%
\appendix
\section{Appendix}

\subsection{Proof for Lemma~\ref{lemma_F_T}} \label{proof_F_T}
Note that $\mathcal{F}(\T)$ approaches $-\infty$ at both end points.
Now, the first and second derivatives of $\mathcal{F}(\T)$ with respect to $\T$ 
are evaluated as:
\begin{equation*}\label{eqn:first_der_F}
	\vspace{-1mm}
\mathcal{F}'(\T) =\frac{1}{\T} - \frac{\alpha}{1-\alpha\T} + 1 , \hspace{3mm} 
\mathcal{F}''(\T) =-\frac{1}{\T^2} - \frac{\alpha^2}{(1-\alpha\T)^2} 
\vspace{-0.25mm}
\end{equation*}
%
Since $\mathcal{F}''(\T)<0$ for $\T>0$, this function is concave.
By setting $\mathcal{F}'(\T)=0$, we can show that the maxima of 
$\mathcal{F}(\T)$ occurs when $\T_m = 
\frac{1}{2\alpha}-1\pm\sqrt{\frac{1}{4\alpha^2}+1}$.
Since the temperature is positive, the following relations hold at the maximum 
point:
\begin{equation} \label{appendix_eqn:inflection_point}
\vspace{-1mm}
\T_m = \frac{1}{2\alpha}-1+\sqrt{\frac{1}{4\alpha^2}+1},~i.e., \alpha = \frac{\T_m+1}{\T_m{^2}+2 \T_m}
\end{equation}
%
Moreover, due to the concavity of $\mathcal{F}(\T)$ is an increasing function 
on $(0,\T_m)$ and decreasing function on $(\T_m,\frac{1}{\alpha})$,
as depicted in Figure~\ref{fixed_point_illustration}.


\subsection{Proof for Theorem~\ref{thm_fixed_point}} \label{existence_proof}

\begin{proof}
Function $\mathcal{F}$ reaches its peak at $\T=\T_m$,
and its solution contains two points if and only if $\mathcal{F}(\T_m)\geq0$ 
(Figure~\ref{fixed_point_illustration}(b)).
Otherwise, if $\mathcal{F}(\T_m)<0$, it does not intersect the x-axis and there 
is no solution (Figure~\ref{fixed_point_illustration}(a)).
The condition $\mathcal{F}(\T_m)\geq0$ is equivalent to:
\begin{align} \label{eq:theorem1}
\mathcal{F}(\T_m) = {} & \ln \beta + \T_m + \ln(\T_m(1-\alpha \T_m)) \nonumber 
\\
 = {} & \ln \beta + \T_m + \ln(\T_m(1-\frac{\T_m+1}{\T_m^2+2 \T_m} \T_m)) 
 \nonumber 
\\
 = {} & \ln \beta + \T_m - \ln\left(\frac{2}{\T_m}+1 \right)\geq0.
\end{align} Hence, $\beta\geq\left(\frac{2}{\T_m}+1\right)e^{-\T_m}$ follows.
\end{proof}

\subsection{Proof for Lemma~\ref{lem_signFT}} \label{proof_lem_sign_FT}
\begin{proof}
The temperature iteration equation is:
\[
T[k+1]=aT[k]+b(P_C+V_{\kappa_1}T[k]^2e^{\frac{\kappa_2}{T[k]}}.
\]
We can rewrite this equation by a change of variable, i.e., 
$T[k]=\frac{-\kappa_2}{\T[k]}$, as:
\[
\frac{-\kappa_2}{\T[k+1]}=-\frac{-a\kappa_2}{\T[k]}+bP_C+bV\kappa_1\frac{\kappa_2^2}{\T[k]^2}e^{-\T[k]}.
\] After substituting the definitions of $\alpha$ and $\beta$ in 
\ref{change_of_parameters} and rearrangement of terms, we obtain:
\begin{align*}
\frac{1}{\T[k+1]} = {} & 
\frac{a}{\T[k]}-\frac{bP_C}{\kappa_2}-\frac{bV\kappa_1\kappa_2}{\T[k]^2}e^{-\T[k]},
 \\
= {} & \frac{a}{\T[k]}-(a-1)\alpha-\frac{(a-1)}{\beta\T[k]^2}e^{-\T[k]}, \\
= {} & 
\frac{1}{\T[k]}+\frac{a-1}{\T[k]}-(a-1)\alpha-\frac{(a-1)}{\beta\T[k]^2}e^{-\T[k]},
 \\
= {} & \frac{1}{\T[k]}-\frac{(1-a)}{\beta 
\T[k]^2}\left(\beta(1-\alpha\T[k])\T[k]-e^{-\T[k]}\right).
\end{align*}
Note that when $\mathcal{F}(\T)<0$,  we can show that $\beta (1-\alpha 
\T)\T-e^{-\T}<0$ using Equation~\ref{eqn:fixed_point_ln}. By inspecting the 
last equation, this means that each temperature iteration decreases 
$\T[k]$ when $\mathcal{F}(\T)$ is negative. In contrast, $\mathcal{F}(\T)>0$ 
implies that  $\beta (1-\alpha \T)\T-e^{-\T}>0$. That is, each fixed point 
iteration increases $\T[k]$ when $\mathcal{F}(\T)$ is positive.
\end{proof}

\subsection{Proof for Theorem~\ref{thm_stability} } \label{stability_proof}
\begin{proof}
When Equation \ref{eqn:fixed_point} has no solution, Equation \ref{eqn:fixed_point_ln} has no solution and the sign of the function $\mathcal{F}(\T)$ is negative since its terms are concave and affine functions. By Lemma \ref{lem_signFT}, at each fixed point iteration, the value of $\T[k]$ decreases towards to $0$. Due to the fact that $T[k] = -\kappa_2 / \T[k] \rightarrow \infty$, thermal runaway occurs as illustrated with the arrows in Figure~\ref{fixed_point_illustration}.

When Equation \ref{eqn:fixed_point} (equivalently Equation 
\ref{eqn:fixed_point_ln}) has a solution, there are two fixed points of the 
function $\mathcal{F}(\T)$, i.e., $\T_u$ and $\T_s$ such that 
$0<\T_u<\T_s<\frac{1}{\alpha}$. Since $\mathcal{F}(\T)$ is a concave down 
function, the sign of this function at $\T\in(0,\T_u)$ is negative and by 
Lemma~\ref{lem_signFT}, the value of $\T$ decreases to $0$ on $\T\in(0,\T_u)$ 
just like the no-solution case, and results in thermal runaway. On the other 
hand, the sign of the function $\mathcal{F}(\T)$ is positive on 
$\T\in(\T_u,\T_s)$ and negative on $\T\in(\T_s,\frac{1}{\alpha})$. By 
Lemma~\ref{lem_signFT}, the value of $\T$ increases in $\T\in(\T_u,\T_s)$ and 
decreases in $\T\in(\T_s,\frac{1}{\alpha})$ towards to $\T_s$ on both 
intervals. Hence, we conclude that $\T_s$ is a stable and $\T_u$ is unstable 
fixed point.
\end{proof}

\end{document}